\documentclass[lettersize,onecolumn]{IEEEtran}
\usepackage{amsmath,amsfonts}
\usepackage{algorithmic}
\usepackage{algorithm}
\usepackage{array}
\usepackage[caption=false,font=normalsize,labelfont=sf,textfont=sf]{subfig}
\usepackage{textcomp}
\usepackage{stfloats}
\usepackage{url}
\usepackage{verbatim}
\usepackage{graphicx}
\usepackage{cite}

\usepackage{amsmath}
\usepackage{amssymb}
\usepackage{mathtools}
\usepackage{amsthm}
 \usepackage{booktabs}

\usepackage{microtype}
\usepackage{mathrsfs}
\usepackage{multicol} % 
\usepackage{hyperref}
\usepackage[capitalize,noabbrev]{cleveref}

\newtheorem{theorem}{Theorem}
\newtheorem{proposition}{Proposition}
\newtheorem{lemma}{Lemma}

\theoremstyle{definition}
\newtheorem{definition}{Definition}
\newtheorem{assumption}{Assumption}
\theoremstyle{remark}
\newtheorem{remark}{Remark}

\newcommand{\argmin}{\operatorname{argmin}}

\newcommand{\mcal}{\mathcal}
\newcommand{\scr}{\mathscr}
\newcommand{\bb}{\mathbb}
\newcommand{\bd}{\boldsymbol}

\newcommand{\optr}{\operatorname}
\newcommand{\la}{\langle}
\newcommand{\ra}{\rangle}
\newcommand{\wtd}{\widetilde}
\newcommand{\wht}{\widehat}

\begin{document}

\title{Online Riemannian Gradient Descent for Quantum State Tomography with Matrix Product Operators}

\author{

\IEEEauthorblockN{
Jian-Feng Cai$^{1,2}$,
Jingyang Li$^{3}$,
Xiaoqun Zhang$^{4,5}$,
Yuanwei Zhang$^{4}$
}

\vspace{3mm}
\IEEEauthorblockA{
$^{1}$Department of Mathematics, Hong Kong University of Science and Technology, Hong Kong\\
$^{2}$IAS Center for AI for Scientific Discoveries, Hong Kong University of Science and Technology, Hong Kong\\
$^{3}$ Department of Statistics and Data Science, Fudan University, Shanghai, China\\
$^{4}$ School of Mathematical Sciences, Shanghai Jiao Tong University, Shanghai, China.\\
$^{5}$ Institute of Natural Sciences, Shanghai Jiao Tong University, Shanghai, China.\\
}

\thanks{
\begin{tabular}[t]{@{}l@{}}
jfcai@ust.hk\\
jjyyli.acad@gmail.com\\
xqzhang@sjtu.edu.cn\\
sjtuzyw@sjtu.edu.cn
\end{tabular}
}
}

\maketitle

\begin{abstract}
Matrix product operators (MPOs) provide a scalable approach for quantum state tomography (QST) by offering a compact representation of many-body mixed states with limited entanglement, using only a number of parameters that scales polynomially with the system size. In this paper, we study QST for quantum density matrices that can be represented by MPOs. We first derive an equivalent characterization of Hermiticity in terms of the MPO core tensors and show that the coefficient tensor of an MPO under the Pauli or generalized Gell-Mann basis admits a real-valued low tensor-train (TT) rank structure. This establishes an explicit connection between MPO-based QST and noisy low-rank tensor completion. Motivated by this formulation, we develop an online Riemannian gradient descent (oRGD) algorithm that sequentially incorporates measurement data during the reconstruction process. With a proper initialization, we prove that oRGD converges linearly to the target MPO and succeeds with a number of distinct measurement settings that scales quadratically with the system size. As a byproduct, our analysis also yields a significantly improved sample complexity bound for the low TT rank tensor completion task. Furthermore, we propose a tailored spectral initialization method and establish its theoretical guarantee. Numerical experiments on several classes of quantum states validate the effectiveness and scalability of the proposed method.
\end{abstract}


\section{Introduction}
\subsection{Background}

Quantum State Tomography (QST) is a fundamental tool in quantum information science for characterizing and validating quantum states produced from experiments. By reconstructing the underlying density matrix from measurement outcomes, QST provides a complete mathematical description of a quantum system and plays a crucial role in the verification of quantum devices, the benchmarking of quantum simulators, and the validation of quantum algorithms \cite{steffen2006measurement,baur2012benchmarking,huang2021benchmarking}. 

Consider an $n$-partite $d$-level quantum system (where $d=2$ corresponds to qubits).
The state of such a system is fully described by a density matrix $\bd{\rho}\in \bb{C}^{d^n\times d^n}$, which is a Hermitian, positive semi-definite matrix with unit trace.
In practice, measurements on this system are mathematically represented by Hermitian observables, such as multi-qubit Pauli operators. Because quantum measurements are inherently probabilistic and destructive, a single measurement of an observable $\bd{A}_k$ on the state $\bd{\rho}$ yields only a discrete eigenvalue. For instance, a Pauli observable on a qubit system yields either $+1$ or $-1$. The theoretical expectation value of these discrete outcomes is given by the Hilbert-Schmidt inner product $\langle \bd{A}_k, \bd{\rho}  \rangle$. 
To accurately estimate this expectation value, one must prepare many identical copies of $\bd{\rho}$ and compute the empirical average over repeated measurement shots. 
% The total number of required copies is referred to as the \emph{sample complexity} of QST. 
In recent years, the scale of quantum computers has increased rapidly \cite{arute2019quantum, chow2021ibm}, making full tomography increasingly challenging as the size of $\bd{\rho}$ increases exponentially with system size $n$. 
This leads to two major bottlenecks: first, the experimental measurement overhead, which encompasses both the number of distinct observables required and the repeated physical shots needed per observable, grows exponentially; second, the computational complexity of reconstructing the state from measurement data becomes prohibitively large.

To overcome the curse of dimensionality, a natural strategy is to exploit low-dimensional structures within the quantum state. A prevalent ansatz assumes the target density matrix has a low matrix rank $r$, effectively framing QST as a low-rank matrix recovery problem. Numerous methodologies have been developed under this framework, including maximum likelihood estimation \cite{hradil1997quantum,vrehavcek2001iterative}, least-squares estimators \cite{kyrillidis2018provable,francca2021fast,hsu2024quantum}, and machine learning techniques \cite{torlai2018neural, carleo2019machine}.  However, the number of physical state copies required by these methods for recovery guarantee inherently scales exponentially with the system size $n$.  In particular, \cite{kueng2017low} shows that $O(d^n r^2)$ physical state copies are necessary (and sufficient) to uniquely determine the target quantum state if the rank-$1$ measurements are allowed. This means that even for the pure state ($r=1$), both the sample complexity and computational cost of low-rank matrix approaches scale exponentially with $n$, rendering them impractical for large quantum systems.

To overcome the limitations of matrix rank, an alternative effective ansatz leverages the entanglement properties of the quantum system via tensor networks. For one-dimensional quantum systems, states exhibiting low entanglement admit a highly compact parameterization: pure states can be efficiently represented as Matrix Product States (MPS) \cite{verstraete2006matrix,perez2007matrix}, while mixed states are captured by Matrix Product Operators (MPO) \cite{pirvu2010matrix,guth2020efficient}. In computational mathematics, this representation is also known as the tensor train (TT) decomposition and has been extensively studied \cite{oseledets2011tensor, holtz2012manifolds} (see \cref{subsec: MPO and TT} for more details). An MPS or MPO with a bounded TT rank (referred to as \emph{bond dimension} in physics) involves only $\operatorname{poly} (n)$ parameters, making it efficient for large systems. Various QST approaches have been proposed under the assumption that the target quantum state possesses such a low-dimensional structure \cite{cramer2010efficient,baumgratz2013scalable,lanyon2017efficient,guo2024quantum,tang2025sketch}. For a pure state that can be represented by an injective MPS and possesses a local, gapped parent Hamiltonian, local information (measurements on adjacent qudits up to a fixed distance) is sufficient to reconstruct the entire state \cite{cramer2010efficient}. In this setting, the number of physical state copies of QST is generally understood to scale polynomially with $n$. 

However, for mixed states represented by MPOs, this local reconstructability generally breaks down, making global measurements necessary to guarantee unique and reliable recovery.
Although numerous methods have been proposed for QST of MPOs \cite{baumgratz2013scalable,han2022density,torlai2023quantum,kuzmin2024learning}, theoretical guarantees regarding their convergence and sample complexity remain lacking. Recent work \cite{qin2024quantum} has established the sample complexity of QST for MPOs using Haar random projective measurement, showing that only $O(n^3 d^2 r_{\max}^2/\epsilon^2)$ state copies suffice to uniquely determine the $\epsilon$-precise state of target MPO via least squares minimization. 
However, this result establishes an information-theoretic bound. Finding such a global least-squares estimator is computationally intractable due to the highly non-convex nature of the low-rank tensor manifold, and no efficient algorithm was provided to compute it.
% and \cite{qin2024sample} proposed a projected gradient descent (PGD) algorithm for QST in this context. 
Furthermore, implementing high-precision Haar random projective measurement is experimentally prohibitive in practice. Moreover, because Haar random projective measurements are entangled and require global rotations of the entire qudit system, the reconstruction process in QST algorithm incurs computational costs that scale exponentially with $n$. These limitations render such methods impractical for large quantum systems.

\subsection{Main Results}

In this paper, we consider Pauli measurements for qubits and generalized Gell-Mann measurements for qudits, which are popular choices for QST that are readily implementable in practical experiments \cite{gross2010quantum,riofrio2017experimental}. 
Since both $n$-qubit Pauli matrices and $n$-qudit generalized Gell-Mann matrices are constructed as Kronecker products of $n$ local single-system matrices, the operation between these observable matrices and an MPO can be computed particularly efficiently. To further alleviate the computational burden of processing massive datasets simultaneously, we formulate the QST problem in an online setting, where measurement data arrive sequentially during the reconstruction procedure. This online framework is widely adopted in the literature \cite{zhang2020online,chen2024adaptive,cai2025online} because it allows state reconstruction to proceed concurrently with data acquisition, eliminating the need to wait for the entire measurement dataset.

Because the set of local tensor-product observables forms a complete orthonormal basis, reconstructing an MPO is mathematically equivalent to reconstructing its coefficient tensor in this measurement basis. \emph{As our first main contribution (\cref{subsec: equ hermitian}), we prove that for any Hermitian matrix represented by a low-rank MPO, its corresponding coefficient tensor naturally inherits a real, low tensor-train rank structure.}
% \red{[If the matrix has no structure, still low TT rank?]}
To establish this, we derive a necessary and sufficient condition for Hermiticity directly in terms of the MPO core tensors, and we provide an explicit algorithm to construct such a Hermitian decomposition. By leveraging this structural equivalence, we rigorously cast the QST problem as a noisy low-rank tensor completion task over a real TT manifold.

Building upon this formulation, \emph{our second main contribution introduces an online Riemannian Gradient Descent (oRGD) algorithm tailored for the low TT-rank tensor manifold.} This work marks the first rigorous development and convergence analysis of online Riemannian optimization for TT manifolds. Given a warm initialization, the oRGD algorithm generates iterates whose reconstruction error in the Frobenius norm converges linearly. An informal statement of the main result is as follows:

\begin{theorem}[Informal version of \cref{thm: local convergence}] 

Let $\bd{\rho}^*$ be an MPO state with bond dimension $r_{\max}$ and let $\epsilon>0$. Suppose an initial estimate $\bd{\rho}_0$ satisfies $\|\bd{\rho}_0 - \bd{\rho}^*\|_F\leq C\cdot \lambda_{\min}/ (n^2r_{\max}^2)$ for some constant $C>0$. Then, given $\Omega(n^2 r^2_{\max} \log \frac{1}{\epsilon})$ tensor-product measurements (entries of the coefficient tensor in the measurement basis), the oRGD algorithm outputs a reconstructed density matrix $\bd{\rho}_{\text{rec}}$ such that $\|\bd{\rho}_{\text{rec}} - \bd{\rho}^* \|_F\leq \epsilon$.
\end{theorem}

This theorem demonstrates that oRGD achieves linear convergence with a measurement setting complexity bound of $O(n^2 r_{\max}^2)$. We emphasize that this bound refers specifically to the number of distinct measurement settings (i.e., the unique tensor-product observables selected).
% to the number of distinct measurement settings (i.e., the unique Pauli observables selected). 
This quadratic scaling represents a substantial reduction compared to standard low-rank matrix QST methods \cite{gross2010quantum,kim2023fast,hsu2024quantum,cai2025online} and offline RGD algorithms for TT completion \cite{cai2022provable,bian2025fast}, both of which inherently require an exponentially growing number of observed entries even with a warm initialization. By drastically reducing the required measurement configurations, oRGD significantly alleviates the experimental calibration overhead and entirely avoids the exponential classical processing costs that typically limit exact MPO tomography. 

Furthermore, because the local convergence of oRGD relies on a warm start, we propose a tailored spectral initialization algorithm equipped with rigorous theoretical guarantees. This approach considerably tightens the theoretical measurement setting bounds compared to direct applications of existing tensor completion initialization schemes to the QST setting \cite{cai2022provable}.

\subsection{Notation}
Throughout this manuscript, we adopt the following notations: calligraphic bold letters (e.g., $\bd{\mcal{T}}$) denote tensors, bold capital letters (e.g., $\bd{X}$) denote matrices, bold lowercase letters (e.g., $\bd{y}$) denote vectors, blackboard bold letters (e.g., $\bb{M}, \bb{C}$) denote sets and calligraphic letters (e.g., $\mcal{P}$, $\mcal{F}$) denote linear maps. For a positive integer $n$, let $[n]$ denote the set $\{1,2, \dots, n\}$. We use $\operatorname{Tr}(\bd{M})$ to represent the trace of a square matrix $\bd{M}$. For two matrices $\bd{A}, \bd{B}$, $\bd{A}\otimes \bd{B}$ is the Kronecker (tensor) product. In the quantum context, a pure state is written in Dirac notation as $|\psi\rangle$ (with $\langle \psi|$ its conjugate transpose), and a mixed state is denoted by $\bd{\rho}$. For a complex number $a$, $\overline{a}$ denotes its complex conjugate. The superscripts $(\cdot)^\top$ and $(\cdot)^\dagger$ denote the transpose and conjugate transpose, respectively.

Norms are defined as follows:  $\|\cdot\|_F$ is the Frobenius norm for tensors and matrices. For $0< p\leq \infty$, $\|\cdot\|_{\ell_p}$ is the $\ell_p$-norm of tensors. Specifically, $\|\bd{\mcal{T}}\|_{\infty}$ is the largest absolute entry of $\bd{\mcal{T}}$. For a matrix $\bd{X}$, $\|\bd{X}\|_{\ell_{2, \infty}}:=\max_k \|\bd{X}(k, :)\|_2$ is the $\ell_{2, \infty}$-norm. Given positive integers $r_1, \dotsm, r_{n-1}$, we use $r_{\max}:=\max_k r_k$ and $r_{\min} = \min_{k} r_k$. For two positive quantities $a, b\in \bb{R}$, we write $b=O(a)$ if $b\leq c a$ for an absolute constant $c>0$, and use $b=\Omega(a)$ to represent $b\geq c a$ for an absolute constant $c>0$.

\section{Preliminaries}
In this section, we begin by reviewing basic concepts related to quantum states, density matrices, and tensor-product measurement schemes (e.g., Pauli and generalized Gell-Mann measurements). Then we introduce the matrix product operators, their connection with tensor-train decomposition and corresponding operations related to tensors. For further background, we refer the reader to \cite{nielsen2010quantum} for quantum information and to \cite{oseledets2011tensor} for tensor-train decomposition.

\subsection{Quantum State and Measurement.}\label{subsec: QuantState Measure}

\textit{Quantum state and density matrix.} 
Quantum states can be divided into two types: the pure state and the mixed state. In a $d$-level $n$-body quantum system, the pure state can be represented by a unit-length vector $|\psi\rangle\in \bb{C}^{d^n}$ (using the Dirac notation). A mixed state, on the other hand, is described by a statistical ensemble of pure states $\{(p_i, |\psi_i\rangle)\}$, where $0 \le p_i \le 1$ represents the classical probability of the system being in state $|\psi_i\rangle$, and $\sum_i p_i = 1$. Thus, the mixed state is usually described by a density matrix, which can be written as 
$$
\bd{\rho} = \sum_{i} p_i |\psi_i\rangle\langle \psi_i|\in \bb{C}^{d^n\times d^n}
$$
A density matrix represents a pure state if and only if its rank is exactly one, otherwise it is a mixed state. And a Hermitian matrix $\bd{\rho}$ in $\bb{C}^{d^n\times d^n}$ corresponds to a mixed state as long as it is positive semi-definite and has a unit trace $\optr{Tr}(\bd{\rho}) = 1$.

\textit{Pauli matrix.}
For qubit systems ($d=2$), let $\{\bd{S}_X, \bd{S}_Y, \bd{S}_Z\}$ be the single-qubit Pauli matrices, and we write $\{\bd{P}_1, \bd{P}_2, \bd{P}_3, \bd{P}_4\} = \{\frac{\sqrt{2}}{2}\bd{I}, \frac{\sqrt{2}}{2}\bd{S}_X, \frac{\sqrt{2}}{2}\bd{S}_Y, \frac{\sqrt{2}}{2}\bd{S}_Z\}$ which forms a complete orthonormal basis in $\bb{C}^{2\times 2}$ under the Hilbert-Schmidt inner product.

\emph{Generalized Gell-Mann matrix.} The generalized Gell-Mann matrices (GGMs) extend the Pauli matrices (for qubit $d=2$), and the Gell-Mann matrices (for qutrits $d=3$) to arbitrary dimension $d$ \cite{kimura2003bloch,bertlmann2005optimal}. For a single qudit, there are $d^2$ such matrices, comprising the identity, $(d-1)$ diagonal GGM, $d(d-1)/2$ symmetric GGM and $d(d-1)/2$ antisymmetric GGM; explicit constructions can be found in \cite{bertlmann2008bloch}. These matrices are Hermitian and we denote those scaled GGM matrices as $\{\bd{P}_1,\bd{P}_2,\dots,\bd{P}_{d^2}\}$, which form a complete orthonormal basis in $\bb{C}^{d\times d}$ under the Hilbert-Schmidt inner product. 

\emph{Practical quantum state tomography and statistical noise.}
Consequently, for a general $n$-qudit system ($d \ge 2$), any tensor-product observable is defined as the Kronecker product of $n$ single-qudit matrices. We denote the complete orthonormal set of these observables as
$\mathbb{W} := \{\boldsymbol{A}_s = \boldsymbol{P}_{s_1} \otimes \boldsymbol{P}_{s_2} \otimes \cdots \otimes \boldsymbol{P}_{s_n} \mid s_k \in \{1, \dots, d^2\}, k=1,\dots,n\}$. Any $n$-qudit density matrix can be expanded in this measurement basis. 

In practical QST, measuring a chosen observable $\boldsymbol{A}_s \in \mathbb{W}$ on a single copy of $\boldsymbol{\rho}$ yields a discrete eigenvalue outcome. To reliably estimate the theoretical expectation value $\langle \boldsymbol{A}_s, \boldsymbol{\rho} \rangle$, one must prepare $M$ identical physical state copies of $\boldsymbol{\rho}$ and compute the empirical average of the $M$ independent measurement outcomes. This empirical average serves as our practical measurement $y_s$, which is inevitably corrupted by quantum statistical fluctuations:
$$y_s = \langle \boldsymbol{A}_s, \boldsymbol{\rho} \rangle + z_s,$$
where $z_s \in \mathbb{R}$ represents the statistical noise. Because $y_s$ is the average of bounded independent random variables, $z_s$ satisfies standard concentration bounds. Specifically, when focusing on the $n$-qubit case ($d=2$), the scaling factor ensures the spectral norm of each scaled Pauli observable is exactly $2^{-n/2}$, leading to the following precise bound:
\begin{lemma}\label{lemma:variance}
    For a scaled Pauli observable $\bd{A}_s$, the statistical noise $z_s$ is a mean-zero random variable with $|z_s|\leq 2^{1-\frac{n}{2}}$, satisfying
    $$
    \optr{Var}(z_s)\leq \frac{1}{2^n M},\quad \text{and}\quad \bb{P}(|z_s|\geq \xi )\leq 2e^{- \xi^2 M 2^{n-1}}, \xi\geq 0
    $$
    where $M$ is the number of physical state copies of $\bd{\rho}$ used to measure $\bd{A}_s$.  
\end{lemma}
The goal of quantum state tomography is to reconstruct the density matrix $\boldsymbol{\rho}$ from the collected measurements $\{(\boldsymbol{A}_s, y_s)\}$, and to determine the number of distinct measurement settings required for provably accurate reconstruction.

\subsection{Matrix Product Operator and Related Tensor Operations.}\label{subsec: MPO and TT}

The dimension of the general quantum state represented by its density matrix $\bd{\rho}$ grows exponentially with particle numbers $n$, leading to the curse of dimensionality in both computational cost and sample requirements in QST \cite{guctua2020fast, gross2010quantum}. The MPO ansatz offers a tractable representation for breaking this barrier for large-scale quantum systems \cite{verstraete2004matrix}. Concretely, a density matrix is called an MPO with bond dimension $\bd{r} = (r_1, \dots, r_{n-1})$, if it admits a matrix product form of each entry:
\begin{equation}\label{equ: MPO}
    \bd{\rho}\left((i_1, \dots, i_n), (j_1, \dots, j_n)\right) = \sum_{l_1, \dots, l_{n-1}}\bd{\mcal{U}}_1(i_1, j_1, l_1) \bd{\mcal{U}}_2(l_1, i_2, j_2, l_2) \cdots \bd{\mcal{U}}_{n}(l_{n-1}, i_n, j_n)
\end{equation}
where $\bd{\mcal{U}}_{i}\in \mathbb{C}^{r_{i-1}\times d\times d\times r_{i}}, i=1,\dots, n$ are called core tensors (with the convention $r_0 = r_n = 1$). For brevity, such a decomposition is denoted by $\bd{\rho} = [\bd{\mcal{U}}_1, \bd{\mcal{U}}_2, \dots, \bd{\mcal{U}}_n]$. In the MPO format \eqref{equ: MPO}, the variable dimension of $\bd{\rho}$ is $O(n d^2 r_{\max}^2)$, which only scales linearly with $n$. This concise representation enables efficient computations. For example, computing the inner product of an MPO of $d=2$ with a Pauli observable reduces to contracting each $2\times 2$ Pauli matrix with core tensor, followed by a sequence of small matrix multiplications. Recent studies demonstrate that density matrix with an MPO structure can be uniquely determined with $O(\optr{poly} n)$ random Haar measurements\cite{qin2024quantum}, showing the efficiency of the MPO ansatz in reducing the measurement cost of QST.

\textit{Connection with tensor train decomposition}. The form in \eqref{equ: MPO} is equivalent to a tensor train decomposition \cite{oseledets2011tensor} if we reshape the $\bd{\rho}$ into an $n$-th order tensor $\bd{\mcal{T}}$ with size $d^2\times \dots \times d^2$ by combining the $(i_k, j_k)$ index pair into a single index $s_k =  i_k + d(j_k-1), k=1, \dots, n$. 
\begin{equation}\label{equ: tensor-train decomposition}
\bd{\mcal{T}}(s_1, \dots, s_n) = \sum_{l_1, \dots, l_{n-1}}\bd{T}_1(s_1, l_1)\bd{T}_2(l_1, s_2, l_2)\cdots \bd{T}_n(l_{n-1}, s_n)
\end{equation}
where $\bd{T}_k(l_{k-1}, s_k, l_k) = \bd{\mcal{U}}_k(l_{k-1}, i_k, j_k, l_k), k=1, \dots, n$ are the reshaped core tensors. We use $\bd{\mcal{T}} = [\bd{T}_1, \dots, \bd{T}_n]$ to denote a tensor in tensor train format. Here we slightly abuse the notation by writing the core tensors with a capital letter $\bd{T}_k$, in order to distinguish them from tensors $\bd{\mcal{U}}_k$ used in the MPO representation.

\textit{Separation and condition number.}
For a given tensor train $\bd{\mcal{T}}$, the \emph{$k$-th separation} of $\bd{\mcal{T}}$ is an $d^{2k}\times d^{2(n-k)}$ matrix $\bd{\mcal{T}}^{\la k\ra}$ defined by
$$
\bd{\mcal{T}}^{\la k\ra}((s_1, \dots, s_k), (s_{k+1}, \dots, s_n)) = \bd{\mcal{T}}(s_1, \dots, s_n).
$$
We define the \emph{$k$-th left part} as the matrix $\bd{T}^{\leq k}$ of size $(d^{2k}\times r_k)$ with entries
$$\bd{T}^{\leq k}((s_1, \dots, s_k), l_k) = \sum_{l_1, \dots, l_{k-1}}\bd{T}_1(s_1, l_1)\bd{T}_2(l_1, s_2, l_2)\cdots \bd{T}_k(l_{k-1}, s_k, l_k).$$
Similarly, the \emph{$k+1$-th right part} is the matrix $\bd{T}^{\geq k+1}$ of size $r_{k}\times (d^{2(n-k)})$. Then, the $k$-th separation of $\bd{\mcal{T}}$ admits the factorization $\bd{\mcal{T}}^{\langle k\rangle} = \bd{T}^{\leq k}\bd{T}^{\geq k+1}$ and the TT rank $r_k$ is the rank of separation matrix $\bd{\mcal{T}}^{\langle k\rangle}$. Furthermore, the smallest singular value of $\bd{\mcal{T}}$ is defined as $\lambda_{\min}(\bd{\mcal{T}}) = \min_{1\leq k\leq n-1} \lambda_{r_k}(\bd{\mcal{T}}^{\la k\ra})$, while the largest singular value is $\lambda_{\max}(\bd{\mcal{T}}) = \max_{1\leq k\leq n-1} \lambda_{1}(\bd{\mcal{T}}^{\la k\ra})$. The condition number is defined as $\kappa(\bd{\mcal{T}}) = \lambda_{\max}(\bd{\mcal{T}})/\lambda_{\min}(\bd{\mcal{T}})$.

\textit{Left and right unfoldings.} It is well known that TT decomposition in \eqref{equ: tensor-train decomposition} for a given tensor is not unique. For identifiability, we first introduce left and right unfoldings of $3$rd order tensor. For any $\bd{\mcal{A}}\in \bb{C}^{d_1\times d_2\times d_3}$, the \emph{left unfolding} operator $L: \mathbb{C}^{d_1\times d_2\times d_3}\longrightarrow \bb{C}^{d_1d_2\times d_3}$ is defined by $L(\bd{\mcal{A}})(s_1 + d_1(s_2 - 1), s_3) = \bd{\mcal{A}}(s_1, s_2, s_3)$, which amounts to a suitable reshaping of the entries. Analogously, the \emph{right unfolding} $R(\bd{\mcal{A}})\in \bb{C}^{d_1\times d_2d_3}$ can be defined similarly. We say the core tensor $\bd{T}_k$ is \emph{left-orthogonal} if $L(\bd{T}_k)^\dagger L(\bd{T}_k)$ is an identity matrix. For identifiability of the TT-format core tensors, we shall always assume that $\bd{T}_1, \dots, \bd{T}_{n-1}$ in \eqref{equ: tensor-train decomposition} are all left-orthogonal. Such a decomposition is called a left-orthogonal decomposition of $\bd{\mcal{T}}$. 

\textit{Tensor train singular value decomposition (TTSVD).} A left-orthogonal decomposition of $\bd{\mcal{T}}$ can be obtained via the TTSVD procedure (\cref{alg: TTSVD}), originally proposed by \cite{oseledets2011tensor}. Here we adopt the restated formulation given in \cite{cai2022provable}. It should be noted that the output low rank tensor from \cref{alg: TTSVD} is generally not the best low TT rank $\bd{r}$ approximation of $\bd{\mcal{T}}$. Indeed, finding the best low-rank approximation of an arbitrary tensor is generally NP-hard. Nevertheless, TTSVD remains a widely used method for low-rank TT approximation because it provides a quasi-optimal approximation whose error is controllable.
\begin{center}
\begin{minipage}{0.8\linewidth}
\begin{algorithm}[H]
\caption{TTSVD}\label{alg: TTSVD}
\begin{algorithmic} 
\STATE \textbf{Input:} Arbitrary $\bd{\mcal{T}}\in \bb{C}^{d_1\times\cdots\times d_n}$ and target TT rank $\bd{r} = (r_1, \dots, r_{n-1})$.
\STATE Set $\wht{\bd{T}}^{\leq 0} = 1$.
\FOR{$k = 1, \dots, n-1$}
\STATE Let $L(\wht{\bd{T}}_k)$ be the leading $r_k$ left singular vector of the matrix $(\wht{\bd{T}}^{\leq k-1}\otimes \bd{I}_{d_k})^\dagger \bd{\mcal{T}}^{\la k \ra}$\\
\STATE Set $\wht{\bd{T}}^{\leq k} = (\wht{\bd{T}}^{\leq k-1}\otimes \bd{I}_{d_k}) L(\wht{\bd{T}}_k)$. \\
\ENDFOR
\STATE $\wht{\bd{T}}_n = (\wht{\bd{T}}^{\leq n-1})^\dagger \bd{\mcal{T}}^{\la n-1\ra}$.\\
\STATE \textbf{Output:} $\operatorname{TTSVD}_{\bd{r}}(\bd{\mcal{T}}) = [\wht{\bd{T}}_1, \wht{\bd{T}}_2, \dots, \wht{\bd{T}}_n]$.
\end{algorithmic} 
\end{algorithm}
\end{minipage}
\end{center}

\section{Bridging Quantum State Tomography and Tensor Completion}
For a density matrix with low-rank prior, the Pauli sensing map has been shown to satisfy the restricted isometry property (RIP) \cite{liu2011universal} with high probability when $m = O(r n^62^n)$ distinct measurement settings are sampled, where $r$ denotes the rank of target density matrix. Leveraging the RIP, extensive literature has formulated QST as a low-rank matrix recovery problem from the compressed-sensing perspective, establishing efficient algorithms and theoretical guarantees \cite{haah2016sample,kueng2017low,kyrillidis2018provable,gross2010quantum,hsu2024quantum,cai2025online}. However, under the MPO structure, this property no longer holds. In this section, we first discuss how physical constraints can be incorporated within the MPO structure, and then model the QST problem under tensor-product measurements as a low-rank tensor completion problem endowed with a real tensor-train structure.

\subsection{Equivalent Condition of Hermitian Property}\label{subsec: equ hermitian}

A matrix product operator $\bd{\rho}$ is called physical if it satisfies three conditions: Hermiticity, positive semidefiniteness, and unit trace. Among these, the Hermitian condition plays the most essential role from the perspective of parameter counting, since it effectively halves the number of real degrees of freedom. Therefore, we primarily focus on how the Hermitian property can be enforced directly within the MPO structure. The following \cref{thm: equal Herm} gives a necessary and sufficient condition for an MPO to be Hermitian, expressed directly in terms of its core tensors.

\begin{theorem}[Equivalence of Hermitian Property]\label{thm: equal Herm}
A matrix product operator $\bd{\rho}\in \bb{C}^{d^n\times d^n}$ with bond dimensions $\bd{r}=(r_1, \dots, r_{n-1})$ is Hermitian if and only if it can be represented in the form $\bd{\rho}=[\bd{\mcal{U}}_1, \bd{\mcal{U}}_2, \dots, \bd{\mcal{U}}_n]$, and the core tensors $\bd{\mcal{U}}_k\in \bb{C}^{r_{k-1}\times d\times d\times r_k}$ satisfy the following condition
\begin{equation}\label{equ: equal hermitian condition}
    \bd{\mcal{U}}_k (l_{k-1}, i_k, j_k, l_{k}) = \overline{\bd{\mcal{U}}_k(l_{k-1}, j_k, i_k, l_{k})},\ k=1,\dots,  n,
\end{equation}
where $i_k, j_k=1, \dots, d$ and $l_k =1, \dots, r_k$.
\end{theorem}
\begin{remark}
    It is important to note that the condition in \eqref{equ: equal hermitian condition} is nontrivial. For example, if one multiplies $\bd{\mcal{U}}_k$ by the imaginary unit and $\bd{\mcal{U}}_{k+1}$ by negative imaginary unit, the overall operator $\bd{\rho}$ remains unchanged. However, the resulting tensors $\bd{\mcal{U}}_k$ and $\bd{\mcal{U}}_{k+1}$ will generally not satisfy \eqref{equ: equal hermitian condition}. In addition, the condition \eqref{equ: equal hermitian condition} implies a reduction by half in the number of real parameters needed for each core tensor, which is consistent with the reduction of degrees of freedom imposed by the Hermitian constraint.
\end{remark}
The proof of \cref{thm: equal Herm} is provided in \cref{proof: equal Herm}. Since the MPO (tensor train) decomposition is not unique, we present an explicit algorithm for obtaining such a decomposition for each Hermitian MPO in \cref{alg: HermTTSVD}. Moreover, the conditions in \eqref{equ: equal hermitian condition} remain invariant under real invertible gauge transformations inserted between adjacent core tensors.
\begin{proposition}
   Let $\bd{\rho} = [\bd{\mcal{U}}_1, \dots, \bd{\mcal{U}}_n]$ be a Hermitian MPO decomposition satisfying \eqref{equ: equal hermitian condition}. For any sequence of real invertible matrices $\bd{G}_k\in \bb{R}^{r_{k}\times r_{k}}, k=0, 1,\dots, n$ with $\bd{G}_0 = \bd{G}_n = 1$, define a new decomposition $\bd{\rho} = [\wht{\bd{\mcal{U}}}_1, \dots, \wht{\bd{\mcal{U}}}_n]$ by
    $$
    \wht{\bd{\mcal{U}}}_k(l_{k-1}, i_k, j_k, l_{k}) = \sum_{s_{k-1}, s_{k}}\bd{G}_{k-1}^{-1}(l_{k-1}, s_{k-1})\bd{\mcal{U}}_k(s_{k-1}, i_k, j_k, s_k) \bd{G}_{k}(s_k, l_k), 
    $$
    then the transformed core tensors $\wht{\bd{\mcal{U}}}_1, \dots, \wht{\bd{\mcal{U}}}_n$ also satisfy \eqref{equ: equal hermitian condition}.
\end{proposition}
Regarding the unit trace property, the trace of $\bd{\rho}$ equals the inner product $\la\bd{\rho}, \bd{I}\ra$, which can be computed efficiently by contracting the core tensors with identity matrices. In contrast, enforcing the PSD constraint directly on the core tensors is more challenging. Although \cite{verstraete2004matrix} imposed additional structures on the core tensors to ensure $\bd{\rho}$ is PSD, this condition is only sufficient and not necessary. Moreover, incorporating both trace and PSD constraints does not substantially reduce the number of degrees of freedom in the MPO representation. Hence, their impact on improving the required measurement setting complexity is likely marginal. For these reasons, similar to other related work \cite{qin2024quantum,qin2024sample}, we primarily focus on the Hermitian property and the low-rank MPO structure, leaving the full characterization of PSD constraints to future investigation.
\subsection{Low Rank Tensor Completion Modelling}
\begin{figure}[t]
    \centering
    \includegraphics[width=0.8\linewidth]{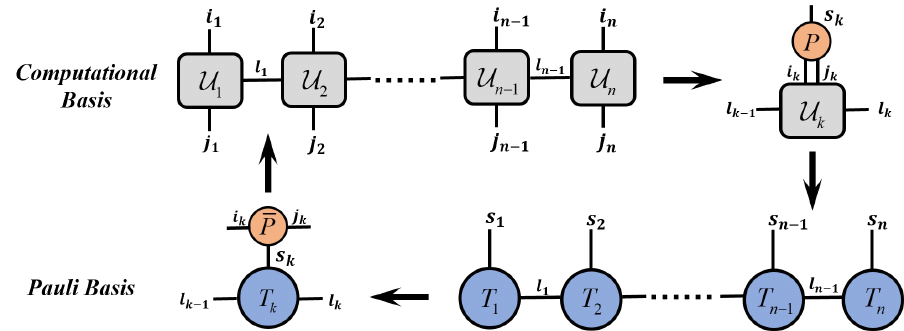}
    \caption{Tensor network diagram connecting the matrix product operator representation of a density matrix to its corresponding Pauli coefficient tensor. The transformation between the computational basis and Pauli basis only involves the contraction of the Pauli matrices with core tensors.}
    \label{fig:illustration}
\end{figure}
Since the set of tensor-product observables $\mathbb{W}$ forms a complete orthonormal basis for $\mathbb{C}^{d^n \times d^n}$, any density matrix can be uniquely represented by its coefficients in this basis. This leads to the following definition.
\begin{definition}[Coefficient Tensor]
For an $n$-partite density matrix $\boldsymbol{\rho}$, its coefficient tensor is an $n$-th order real tensor $\bd{\mcal{T}}\in \mathbb{R}^{d^2 \times \cdots \times d^2}$ defined by
\begin{equation}\label{coefficient tensor}
\bd{\mcal{T}}(s_1, \dots, s_n) = \langle \boldsymbol{A}_s, \boldsymbol{\rho} \rangle,
\end{equation}
where $\boldsymbol{A}_s = \boldsymbol{P}_{s_1} \otimes \cdots \otimes \boldsymbol{P}_{s_n} \in \mathbb{W}$ and each index $s_k \in \{1, \dots, d^2\}$.
\end{definition}
The coefficient tensor is a real tensor due to the Hermiticity of both $\bd{\rho}$ and the basis matrices. Furthermore, if the target state is a Hermitian MPO $\bd{\rho}^* = [\bd{\mcal{U}}^*_1, \bd{\mcal{U}}^*_2, \dots,\bd{\mcal{U}}^*_n]$ satisfying condition \eqref{equ: equal hermitian condition}, then the corresponding coefficient tensor admits a \emph{real} tensor train decomposition of the form
\begin{equation}\label{equ: QST to LRTC}
\begin{aligned}
&\bd{\mcal{T}}^*(s_1, \dots, s_n) = \sum_{i_1, \dots, i_n} \sum_{j_1, \dots, j_n} \bd{\rho}^*((i_1, \dots, i_n), (j_1, \dots, j_n)) \bd{P}_{s_1}(i_1, j_1)\cdots \bd{P}_{s_n}(i_n, j_n)\\
&\quad = \sum_{l_1, \dots, l_{n-1}} \left(\sum_{i_1, j_1} \bd{\mcal{U}}^*_1(i_1, j_1, l_1) \bd{P}_{s_1}(i_1, j_1)\right) \left(\sum_{i_2, j_2} \bd{\mcal{U}}^*_2(l_1, i_2, j_2, l_2) \bd{P}_{s_2}(i_2, j_2)\right)\cdots \left(\sum_{i_n, j_n} \bd{\mcal{U}}^*_n(l_{n-1}, i_n, j_n) \bd{P}_{s_n}(i_n, j_n)\right)\\
&\quad =: \sum_{l_1, \dots, l_{n-1}} \bd{T}_1^*(s_1, l_1) \bd{T}_2^*(l_1, s_2, l_2) \cdots \bd{T}_n^*(l_{n-1}, s_n).
\end{aligned}
\end{equation}
Here, each $\bd{T}_k^*\in \bb{R}^{r_{k-1}\times d^2\times r_k}$ is a real tensor due to \cref{thm: equal Herm} and the Hermiticity of the basis matrices. Moreover, the resulting coefficient tensor $\bd{\mcal{T}}^*$ itself inherits a real tensor-train structure whose TT‑rank coincides with the original bond dimension $\bd{r}$ due to the following proposition.
\begin{proposition}\label{prop: left orth of coeff tensor}
    If $\bd{\mcal{U}}^*_1, \dots,\bd{\mcal{U}}^*_{n-1}$ are left orthogonal, then the corresponding tensors $\bd{T}_1^*, \dots, \bd{T}_{n-1}^*$ obtained via \eqref{equ: QST to LRTC} are also left orthogonal. 
\end{proposition}
\begin{proof}
    We can check the orthogonality by direct derivation as follows:
    $$
    \begin{aligned}
    &\sum_{s_k, l_{k-1}}\overline{\bd{T}_k(l_{k-1}, s_k, l_{k})}\bd{T}_k(l_{k-1}, s_k, \hat{l}_{k})\\
    \qquad &=  \sum_{s_k, l_{k-1}} \sum_{i_k, j_k} \overline{\bd{\mcal{U}}^*_k(l_{k-1}, i_k, j_k, l_{k}) \bd{P}_{s_k}(i_k, j_k)}\sum_{\hat{i}_k, \hat{j}_k} \bd{\mcal{U}}^*_k(l_{k-1}, \hat{i}_k, \hat{j}_k, \hat{l}_{k}) \bd{P}_{s_k}(\hat{i}_k, \hat{j}_k)\\
    \qquad &= \sum_{l_{k-1}} \sum_{i_k, j_k} \sum_{\hat{i}_k, \hat{j}_k} \overline{\bd{\mcal{U}}^*_k(l_{k-1}, i_k, j_k, l_{k})}\bd{\mcal{U}}^*_k(l_{k-1}, \hat{i}_k, \hat{j}_k, \hat{l}_{k})\left(\sum_{s_k}\overline{\bd{P}_{s_k}(i_k, j_k)}\bd{P}_{s_k}(\hat{i}_k, \hat{j}_k)\right)\\
    \qquad &= \sum_{l_{k-1}} \sum_{i_k, j_k} \sum_{\hat{i}_k, \hat{j}_k} \overline{\bd{\mcal{U}}^*_k(l_{k-1}, i_k, j_k, l_{k})}\bd{\mcal{U}}^*_k(l_{k-1}, \hat{i}_k, \hat{j}_k, \hat{l}_{k})\delta_{(i_k, j_k), (\hat{i}_k, \hat{j}_k)}\\
    &= \sum_{l_{k-1}} \sum_{i_k, j_k}  \overline{\bd{\mcal{U}}^*_k(l_{k-1}, i_k, j_k, l_{k})}\bd{\mcal{U}}^*_k(l_{k-1}, i_k, j_k, \hat{l}_{k}) = \delta_{l_{k}, \hat{l}_{k}},
    \end{aligned}
    $$
    where the third equality holds due to the orthogonality of basis matrices and the last equality follows directly from the left-orthogonality of $\bd{\mcal{U}}_k^*$.
\end{proof}
Based on this mathematical equivalence, suppose we sample $m$ distinct measurement settings corresponding to the local observables $\{\boldsymbol{A}_{\omega_1}, \dots, \boldsymbol{A}_{\omega_m}\} \subset \mathbb{W}$. The practical observations obtained from the physical experiment can be written directly as noisy sampled entries of the coefficient tensor. Recalling the measurement model from \cref{subsec: QuantState Measure}, for each $i = 1, \dots, m$, we have:
\begin{align*}
y_{\omega_i} = \langle \boldsymbol{A}_{\omega_i}, \boldsymbol{\rho}^* \rangle + z_{\omega_i} = \bd{\mcal{T}}^*(\omega_i) + z_{\omega_i}.
\end{align*}

Consequently, the QST problem under our local measurement scheme is mathematically equivalent to completing the real coefficient tensor $\bd{\mcal{T}}^*$ from its noisy, partially observed entries, subject to a low TT-rank constraint. 
As illustrated in \cref{fig:illustration}, once an estimate $\bd{\mcal{T}}_{\text{rec}}$ is obtained, the corresponding density matrix $\bd{\rho}_{\text{rec}}$ can be recovered via the inverse transformation in \eqref{equ: QST to LRTC}. Finally, because the chosen measurement basis $\mathbb{W}$ is complete and orthonormal, Parseval's identity guarantees that the reconstruction error in the coefficient tensor space exactly equals the error in the density matrix state space:
$$\|\bd{\rho}_{\text{rec}} - \bd{\rho}^*\|_F = \|\bd{\mcal{T}}_{\text{rec}} - \bd{\mcal{T}}^*\|_F.$$

\subsection{Incoherence Condition in Tensor Completion.}
For the completion problem to be well-posed, it is essential to assume that the information in $\bd{\mcal{T}}^*$ is not concentrated in a few entries but is instead fairly distributed. This property is commonly quantified by the following two notions \cite{yuan2016tensor,cai2023generalized}.
\begin{definition}[Incoherence and Spikiness]
For a tensor $\bd{\mcal{T}}^*\in\mathbb{R}^{d^2 \times \cdots \times d^2}$,
the \emph{spikiness} of $\bd{\mcal{T}}^*$ is defined as
$$
\optr{Spiki}(\bd{\mcal{T}}^*) = \frac{d^{n}\cdot\|\bd{\mcal{T}}^*\|_{\infty}}{\|\bd{\mcal{T}}^*\|_F}.
$$
Let the $k$-th separation of $\bd{\mcal{T}}^*$ be given by the factorization $\bd{\mcal{T}}^{*\la k\ra} = \bd{T}^{*\leq k}\bd{\Lambda}_{k+1}^* \bd{V}_{k+1}^{*\top}$, where $\bd{\Lambda}_{k+1}^*$ is an $r_k\times r_k$ invertible matrix and the factors are both orthogonal $\bd{T}^{*\leq k \top}\bd{T}^{*\leq k} = \bd{V}_{k+1}^{*\top} \bd{V}_{k+1}^{*}  = \bd{I}_{r_k}$. The \emph{incoherence} of tensor $\bd{\mcal{T}}^*$ is defined as 
$$
    \optr{Incoh} (\bd{\mcal{T}}^*) := \max_k \{\optr{Incoh}(\bd{\mcal{T}}^{*\la k\ra})\}
    := \max_k\{(d^{2k}/r_k)^{1/2}\|\bd{T}^{*\leq k}\|_{2, \infty} , (d^{2(n-k)}/r_k)^{1/2}\|\bd{V}_{k+1}^*\|_{2, \infty} \}
$$    
\end{definition}
As established by \cref{lemma: spiki imply incoh} and \cref{lemma: incoh imply spiki}, the incoherence and spikiness of a tensor can be mutually bounded. In this work, we formulate the following assumption concerning the spikiness of $\bd{\mcal{T}}^*$. This condition ensures that the entries of the tensor are sufficiently evenly distributed, enabling a small subset of entries to effectively capture and reveal the structure of the entire tensor.
\begin{assumption}
    Let $\bd{\mcal{T}}^*$ be a real tensor with exact TT-rank $\bd{r}$. There exists a $\nu\geq 0$ such that $\optr{Spiki}(\bd{\mcal{T}}^*)\leq \nu$. 
\end{assumption}
By \cref{lemma: spiki imply incoh}, the assumption above directly implies that the incoherence of $\bd{\mcal{T}}^*$ is also bounded by $\optr{Incoh}(\bd{\mcal{T}}^*)\leq \kappa_0 \nu:=\sqrt{\mu}$, where $\kappa_0$ is the condition number of $\bd{\mcal{T}}^*$.
This assumption is mild and commonly adopted in the tensor completion literature \cite{jain2014provable,cai2019nonconvex,xia2021statistically,cai2022provable,zhangpreconditioned}.

\section{Online Riemannian Gradient Descent for Quantum State Tomography}\label{section: oRGD}
The tensor completion problem has been extensively studied for various tensor structures—such as Tucker, CP, and TT—with considerable work devoted to algorithm design and sample complexity analysis \cite{yuan2016tensor,cai2021nonconvex,tong2022scaling,cai2022provable, li2023online,zhangpreconditioned,zhang2026single}. 
Before proceeding, we make a crucial terminology clarification: in the standard tensor completion literature, a sample refers to a revealed entry of the underlying tensor. Because observing a tensor entry in our QST formulation is mathematically equivalent to performing a specific local measurement, we will use the terms sample and distinct measurement setting interchangeably in the algorithmic and theoretical discussions that follow. This concept must be strictly distinguished from the physical state copies (shots) required to compute the empirical average of each measurement. In particular, \cite{cai2022provable} investigates the Riemannian gradient descent (RGD) method in the offline setting and establishes its convergence theory as well as sample requirements. Nevertheless, directly applying RGD (as well as other existing TT-format completion algorithms) to our problem faces two major limitations. First, the sample complexity required by RGD both for initialization and for local convergence grows exponentially with the tensor order $n$. This fails to fully exploit the low‑dimensional nature of the MPO representation and offers no advantage over low‑rank matrix QST approaches. Second, the per‑iteration computational cost of RGD also scales exponentially in $n$, because the algorithm processes the entire set of measurement data simultaneously to compute each gradient update. 

We therefore consider the online measurement setting in which a batch of measurements—or even a single measurement—is performed sequentially, and only the data from the current batch is used to update the estimate in each round. This reduces the algorithmic per‑iteration computational cost to a polynomial in $n$ by fully exploiting the low‑rank TT structure. In this section, we first review the geometry of the fixed-TT-rank manifold and its corresponding tangent space projection. Then, we present the online Riemannian gradient descent (oRGD) algorithm for the QST problem formulated as a low TT‑rank tensor completion task and establish its convergence theory. Notably, in the noiseless case, oRGD requires only a polynomial number of distinct measurement settings to guarantee local linear convergence.

\subsection{Fixed TT Rank Tensor Manifold.}
For TT rank $\bd{r} = (r_1, \dots, r_{n-1})$, we define $\mathbb{M}_{\bd{r}} =\{\bd{\mcal{T}}\in \bb{R}^{d^2\times \cdots\times d^2}, \bd{\mcal{T}} = [\bd{T}_1, \dots, \bd{T}_n], \bd{T}_k\in \bb{R}^{r_{k-1}\times d^2\times r_k}\}$ as the set of $n$-th order real tensors with TT rank equal to $\bd{r}$. It has been proven that $\bb{M}_{\bd{r}}$ is a smooth manifold with dimension $\overline{\operatorname{dof}} =  d^2\cdot\sum_{k=1}^n r_{k-1} r_k - \sum_{k=1}^{n-1} r_k^2$ \cite{holtz2012manifolds}. For any tensor $\bd{\mcal{T}} = [\bd{T}_1, \dots, \bd{T}_n]\in \bb{M}_{r}$, we denote $\bb{T}$ as the tangent space of $\bd{\mcal{T}}$. As demonstrated in \cite{holtz2012manifolds},
 for any $\bd{\mcal{X}}\in \bb{T}$, there exists a sequence of core tensors $\bd{X}_1, \dots, \bd{X}_n$ with $\bd{X}_k\in \bb{R}^{r_{k-1}\times d^2\times r_k}$ such that $\bd{\mcal{X}}$ can be parametrized as 
  $$\bd{\mcal{X}} = \sum_{k=1}^n\delta \bd{\mcal{X}}_k, \quad \delta\bd{\mcal{X}}_k = [\bd{T}_1, \dots, \bd{T}_{k-1}, \bd{X}_k, \bd{T}_{k+1}, \dots, \bd{T}_n].$$
Moreover, the core tensors $\bd{X}_k$ satisfy that $L(\bd{X}_k)^\top L(\bd{T}_k) = \bd{0}_{r_{k}}$ for $k=1, \dots, n-1$. There is no constraint on the last component $\bd{X}_n$. Such parameterization can be directly used for tangent space projection computation. As indicated in \cite{cai2022provable}, for any $n$-th order tensor $\bd{\mcal{A}}\in \mathbb{R}^{d^2\times \cdots\times d^2}$, the projection $\mcal{P}_{\bb{T}}(\bd{\mcal{A}})$ can also be expressed as $\mcal{P}_{\bb{T}}(\bd{\mcal{A}}) = \sum_{k=1}^n \delta \bd{\mcal{A}}_k$, where $\delta \bd{\mcal{A}}_k = [\bd{T}_1, \dots, \bd{T}_{k-1}, \bd{A}_k,\bd{T}_{k+1}, \dots, \bd{T}_n]$. And the sequence of core tensors $\bd{A}_1, \dots \bd{A}_n$ has the following explicit form
\begin{equation}\label{equ: tgs projection}
L(\bd{A}_k) = \left\{\begin{array}{cc}
   (\bd{I}_{d^2 r_{k-1}} - L(\bd{T}_k)L(\bd{T}_k)^\top)(\bd{T}^{\leq k-1}\otimes \bd{I}_{d^2})^\top\bd{\mcal{A}}^{\langle k \rangle} (\bd{T}^{\geq k+1})^\top (\bd{T}^{\geq k+1} (\bd{T}^{\geq k+1})^\top)^{-1},  & k=1, \dots, n-1 \\
    (\bd{T}^{\leq n-1}\otimes \bd{I}_{d^2})^\top \bd{\mcal{A}}^{\langle n \rangle}, & k=n 
\end{array}\right.
\end{equation}

The Riemannian gradient is obtained by projecting the Euclidean gradient onto the tangent space, \eqref{equ: tgs projection} provides a practical way to compute it. Although the formula in \eqref{equ: tgs projection} involves large matrices such as $\bd{\mcal{A}}^{\la k\ra}$ whose dimension scales exponentially in $n$, the actual computational cost in the completion setting remains low. This is because the Euclidean gradient is sparse, stemming from only a small batch (or even a single) observed entries, and the projection can be carried out efficiently, requiring a computational cost that scales only polynomially with $n$.

\subsection{Online RGD Algorithm}
We consider an online setting where, at each iteration $t$, exactly one distinct basis observable is measured. Equivalently, a single entry of the coefficient tensor at index $\omega_t \in [d^2] \times \dots \times [d^2]$ is uniformly sampled and observed. Let $y_{\omega_t} = \bd{\mathcal{T}}^*(\omega_t) + z_{\omega_t}$ be the raw physical measurement defined in \cref{subsec: QuantState Measure}. To formulate an unbiased stochastic gradient for the tensor completion task, we introduce the scaled indicator tensor $\bd{\mathcal{E}}_t := d^n \bd{e}_{\omega_t}$, where $\bd{e}_{\omega_t}$ is the canonical tensor with $1$ at position $\omega_t$ and zeros elsewhere. Accordingly, we define the scaled algorithmic observation $Y_t := d^n y_{\omega_t} = \langle \bd{\mathcal{E}}_t, \bd{\mathcal{T}}^* \rangle + \epsilon_t$, where $\epsilon_t = d^n z_{\omega_t}$ is the scaled statistical noise with a proxy variance $\sigma^2$. The stochastic loss function at iteration $t$ is then formulated as:
\begin{equation}
\label{equ: online completion}
\ell_t(\bd{\mcal{T}}): = \frac{1}{2}\left(\la \bd{\mcal{E}}_t, \bd{\mcal{T}}\ra - Y_t\right)^2,\quad t=0, 1, \dots, T.
\end{equation}

To update the estimate $\bd{\mathcal{T}}$ at round $t$ while preserving the low TT-rank structure, we perform a Riemannian gradient descent step on the previous estimate $\bd{\mathcal{T}}_t$ to decrease the loss $\ell_t(\bd{\mathcal{T}})$. Each iteration consists of three steps:
    % \vspace{-0.2cm}
\begin{enumerate}
    \item \emph{Riemannian gradient computation.} Compute the Euclidean gradient $\bd{\mcal{G}}_t := \nabla \ell_t(\bd{\mcal{T}}_{t})$ and project it onto $\bd{\mcal{T}}_t$'s tangent space via $\mcal{P}_{\bb{T}_{t}}$.
    % \vspace{-0.2cm}
    \item \emph{Descent in tangent space.} Conduct a descent along the Riemannian gradient in the tangent space $\bb{T}_t$ with step size $\eta$ as $\bd{\mcal{W}}_t = \bd{\mcal{T}}_t - \eta \cdot \mcal{P}_{\bb{T}_t}\bd{\mcal{G}}_t$. 
    % \vspace{-0.2cm}
    \item \emph{Retraction}. The updated estimate $\bd{\mcal{W}}_t$ generally has TT rank larger than the target rank $\bd{r}$. A retraction step projects $\bd{\mcal{W}}_t$ back to the manifold $\bb{M}_{\bd{r}}$. This can be implemented by the TTSVD \cref{alg: TTSVD}.
        % \vspace{-0.2cm}
\end{enumerate}
  The complete online RGD algorithm is summarized in \cref{alg: online RGD}.
\begin{center}
\begin{minipage}{0.8\linewidth}
\begin{algorithm}[H]
\caption{online RGD (oRGD)}\label{alg: online RGD}
\begin{algorithmic} 
\STATE \textbf{Initialization:} $\bd{\mcal{T}}_0\in\bb{M}_{\bd{r}}$ and spikiness parameter $\nu$ \\
                    \FOR{$t = 0, 1, \dots, T_{\max}$}
\STATE $ \bd{\mcal{G}}_t = (\la\bd{\mcal{E}}_t, \bd{\mcal{T}}_t \ra - Y_t)\bd{\mcal{E}}_t$. \\
\STATE $\bd{\mcal{T}}_{t}^{+} = \bd{\mcal{T}}_t - \eta\cdot\mcal{P}_{\mathbb{T}_t}\bd{\mcal{G}}_t$. \\
\STATE $\bd{\mcal{W}}_{t} = \optr{Trim}_{\xi_t}(\bd{\mcal{T}}_t^{+}), \xi_t = \frac{10\|\bd{\mcal{T}}_t^{+}\|_F}{9d^n}\nu$.
\STATE $\bd{\mcal{T}}_{t+1} = \operatorname{TTSVD}_{\bd{r}} (\bd{\mcal{W}}_t)$.
\ENDFOR
\end{algorithmic} 
% \label{alg: PRGD}
\end{algorithm}
\end{minipage}
\end{center}
Note that in \cref{alg: online RGD}, we add an additional procedure before the retraction, termed trimming. This is an element-wise operation, for arbitrary $n$-th order tensor $\bd{\mcal{Z}}\in \bb{R}^{d^2\times \cdots\times d^2}$ and index $\bd{s} = (s_1,\dots, s_n)$:
$$
\optr{Trim}_{\xi}(\bd{\mcal{Z}})(\bd{s}) = \left\{\begin{array}{cc}
  \xi\cdot \optr{sign}(\bd{\mcal{Z}}(\bd{s})),  & \text{if } |\bd{\mcal{Z}}(\bd{s})|\geq \xi, \\
   \bd{\mcal{Z}}(\bd{s}),  & \text{otherwise}.
\end{array}\right.
$$
The trimming operation clips large entries of $\bd{\mcal{Z}}$ and is introduced to ensure that the estimate $\bd{\mcal{T}}_{t+1}$ retains a bounded incoherence. It should be emphasized that this trimming step is employed solely for the convenience of the theoretical analysis. In our numerical experiments, we observe that the standard online RGD algorithm without trimming performs nearly identically to the trimmed version. Therefore, in practice, one can completely skip the trimming step when implementing the oRGD algorithm.

We now state the local convergence property of oRGD. In a small neighborhood of the target tensor $\bd{\mcal{T}}^*$ and under practical conditions, including a suitable signal-to-noise ratio condition (sufficiently precise measurements) and appropriately chosen learning rate, the algorithm exhibits linear convergence with high probability. Recall from \cref{subsec: MPO and TT}, we denote $\lambda_{\min}$ as the smallest singular value of $\bd{\mcal{T}}^*$.

\begin{theorem}[Local convergence of online RGD]\label{thm: local convergence}
    Suppose that the initialization $\bd{\mcal{T}}_0\in \mathbb{M}_{\bd{r}}$ satisfies that $\|\bd{\mcal{T}}_0-\bd{\mcal{T}}^*\|_F\leq c\eta \lambda_{\min}$ for some small constant $c$, $\optr{Incoh}(\bd{\mcal{T}}_0)\leq 2\kappa_0 \sqrt{\mu}$, the step size satisfies $C_0\eta\cdot n^2 \mu^2 \kappa_0^4 r_{\max}^2 d^2 \log d\leq 1$ and the signal to noise ratio satisfies $(\lambda_{\min}/\sigma)^2\geq C_1 n^2 \eta^{-1} \overline{\optr{dof}}$, where $C_0, C_1$ are absolute constants. Then with probability exceeding $1-3Td^{-20}$, for all $t\leq T$, the iterates from \cref{alg: online RGD} guarantee
    $$
    \|\bd{\mcal{T}}_t - \bd{\mcal{T}}^*\|_F^2\leq 2\left(1-\frac{1}{4} \eta\right)^t \|\bd{\mcal{T}}_0 - \bd{\mcal{T}}^*\|_F^2 + C_2 \eta \cdot\overline{\optr{dof}} \sigma^2
    $$
    where $C_2$ is an absolute constant.
\end{theorem}
The proof of \cref{thm: local convergence} is provided in \cref{proof: local convergence}. It implies that oRGD produces an $\epsilon$-accurate estimate in the noiseless case after $T = O(\eta^{-1}\log \epsilon^{-1})$ iterations. Since \cref{thm: local convergence} only requires $\eta^{-1} = O(n^2)$, the total number of iterations, and consequently the total sample number, is merely $O(n^2 \cdot \log \epsilon^{-1})$. This improves the previous results for TT format tensor completion \cite{cai2022provable}, where the offline RGD still needs a sample size that scales exponentially with $n$, even when a warm initialization is provided. In the noisy case, besides the conditions on the step size and the signal-to-noise ratio, the reconstruction error is also affected by the noise level. In particular, the error cannot vanish asymptotically and is instead dominated by a noise-dependent order $O(\eta \overline{\optr{dof}} \sigma^2)$. Consequently, to strictly guarantee a final reconstruction error within a desired $\epsilon$-accuracy, this error naturally imposes a stringent upper bound on the tolerable noise magnitude $\sigma$. In the context of quantum measurements, this mathematical requirement translates directly into the physical cost of the experiment.

\begin{remark}[Measurement settings vs. physical copies]
As clarified earlier, the polynomial sample complexity $T = \mathcal{O}(n^2 \log \epsilon^{-1})$ achieved by oRGD refers strictly to the number of distinct measurement settings (i.e., unique local observables) queried by the algorithm. To rigorously satisfy the signal-to-noise ratio (SNR) requirement of \cref{thm: local convergence}, the statistical noise variance $\sigma^2$ must be sufficiently small. From \cref{lemma:variance}, suppressing this intrinsic quantum fluctuation inherently demands an exponential number of physical state copies. While the total number of physical state copies remains exponential, oRGD effectively decouples the measurement setting complexity from this physical constraint. By drastically reducing the required number of distinct experimental configurations to a mere polynomial in $n$, our method significantly simplifies experimental calibration and control.

\end{remark}

\subsection{Initialization by Sequential Second-Order Spectral Method}

In this section, we propose a second-order spectral method to obtain an initial estimate sufficiently close to the target coefficient tensor $\bd{\mcal{T}}^*$. Such second-order spectral approaches are widely used in tensor learning tasks \cite{xia2019polynomial,cai2022provable}, where sample data are used to construct a Gram matrix, and the resulting eigenvectors provide an approximation of the tensor components for initialization. In particular, \cite{cai2022provable} developed an analysis for the tensor train completion problem. However, the required sample size in \cite{cai2022provable} grows far more rapidly than the desired $O(d^n)$ and can even exceed the total number of entries in the tensor. To address this issue, we introduce a refined initialization scheme tailored to the QST problem that achieves a significantly improved sample complexity. 

\begin{figure}[h]
    \centering
    \includegraphics[width=0.7
    \linewidth]{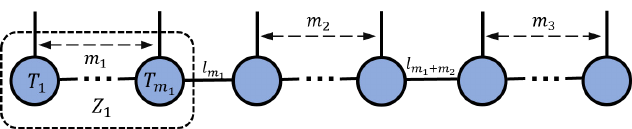}
    \caption{Reshape the $n$-th order tensor $\bd{\mcal{T}}$ into a $3$-rd order tensor $\bd{\mcal{Z}}.$}
    \label{fig:illustration_initial}
\end{figure}

Let $m_1 = \lceil n/3 \rceil, m_2 = \lfloor n/3 \rfloor$ and $m_3 = n-m_1- m_2$. The target tensor $\bd{\mcal{T}}^*$ can be viewed as a third-order tensor of dimension $(d^{2m_1}, d^{2m_2}, d^{2m_3})$ with TT rank $ (r_{m_1}, r_{m_1+m_2})$, as illustrated in \cref{fig:illustration_initial}. We first apply the initialization method from \cite{cai2022provable} using online data to obtain a third-order tensor $\wht{\bd{\mcal{Z}}} = [\wht{\bd{Z}}_1, \wht{\bd{Z}}_2, \wht{\bd{Z}}_3]$ that is sufficiently close to $\bd{\mcal{T}}^*$. Then we use the rank-$\bd{r}$ approximation of $\wht{\bd{\mcal{Z}}}$ as the desired initial estimate. The complete procedure is summarized in \cref{alg: online initializaton}, and its theoretical guarantee is stated in \cref{prop: intialization estimation}.

\begin{center}
\begin{minipage}{0.8\linewidth}
\begin{algorithm}[H]
\caption{online initialization}\label{alg: online initializaton}
\begin{algorithmic} 
\STATE \textbf{Input:} Spikiness parameter $\nu$ and incoherence parameter $\mu$ \\
\STATE Collect data $\{\bd{\mcal{E}}_t, Y_t\}_{t=1}^{2K_1}$ and compute
$$
\bd{N}_1 = \frac{1}{2K_1^2}\sum_{k=1}^{K_1} \sum_{l=K_1 +1}^{2 K_1} Y_k Y_l \left(\bd{\mcal{E}}_k^{\la m_1\ra} \bd{\mcal{E}}_l^{\la m_1\ra^\top} + \bd{\mcal{E}}_l^{\la m_1\ra} \bd{\mcal{E}}_k^{\la m_1\ra^\top}\right).
$$
\STATE Set $\wtd{\bd{Z}}_1$ be the top $r_{m_1}$ left singular vectors of $\bd{N}_1$.
\STATE Truncation: $\bar{\bd{Z}}_1^{i} = \frac{\wtd{\bd{Z}}_1^i}{\|\wtd{\bd{Z}}_1^i\|_2}\cdot \min\{\|\wtd{\bd{Z}}_1^i\|_{2}, \sqrt{\mu r_{m_1}}/d^{m_1} \}$
\STATE Re-normalization: $\wht{\bd{Z}}_1 = \bar{\bd{Z}}_1 (\bar{\bd{Z}}_1^\top \bar{\bd{Z}}_1 )^{-1/2}$.
\STATE Collect data $\{\bd{\mcal{E}}_t, Y_t\}_{t=1}^{2K_2}$ and compute 
$$
\bd{N}_2 = \frac{1}{2K_2^2}\sum_{k=1}^{K_2} \sum_{l=K_2 +1}^{2 K_2} Y_k Y_l \left(\bd{\mcal{E}}_k^{\la m_1+m_2\ra} \bd{\mcal{E}}_l^{\la m_1+m_2\ra^\top} + \bd{\mcal{E}}_l^{\la m_1+m_2\ra} \bd{\mcal{E}}_k^{\la m_1+m_2\ra^\top}\right).
$$
\STATE Set $L(\wtd{\bd{Z}}_2)$ be the top $r_{m_1+m_2}$ left singular vectors of $(\wht{\bd{Z}}_1\otimes \bd{I})^\top \bd{N}_2 (\wht{\bd{Z}}_1\otimes \bd{I})$.
\STATE Truncation: $L(\bar{\bd{Z}}_2)^{i} = \frac{L(\wtd{\bd{Z}}_2)^i}{L(\|\wtd{\bd{Z}}_2)^i\|_2}\cdot \min\{\|L(\wtd{\bd{Z}}_2)^i\|_{2}, \sqrt{\mu r_{m_1 +m_2}}/ d^{m_2} \}$
\STATE Re-normalization: $L(\wht{\bd{Z}}_2) = L(\bar{\bd{Z}}_2) \left(L(\bar{\bd{Z}}_2)^\top L(\bar{\bd{Z}}_2) \right)^{-1/2}$.
\STATE Collect $\{ \bd{\mcal{E}}_t, Y_t\}_{t=1}^{K_3}$ and compute $\wht{\bd{Z}}_3 = \frac{1}{K_3}\left(\wht{\bd{Z}}^{\leq 2}\right)^\top \left(\sum_{k=1}^{K_3} Y_k \bd{\mcal{E}}_k^{\la m_1 + m_2\ra}\right)$
\STATE Reconstruction $\wht{\bd{\mcal{Z}}} = [\wht{\bd{Z}}_1, \wht{\bd{Z}}_2, \wht{\bd{Z}}_3]$.
\STATE \textbf{Output:} $\bd{\mcal{T}}_0 = \optr{TTSVD}_{\bd{r}}(\optr{Trim}_{\xi}(\wht{\bd{\mcal{Z}}}))$ with $\xi = \frac{10 \|\wht{\bd{\mcal{Z}}}\|_F}{9 d^n}\nu$
\end{algorithmic} 
\label{alg: PRGD}
\end{algorithm}
\end{minipage}
\end{center}

\begin{proposition}\label{prop: intialization estimation}
    Suppose $\scr{K}$ is the desired upper bound, then if the sample number $K_1, K_2$ and $K_3$ in \cref{alg: online initializaton} have a lower bound $K$ satisfying
    $$
    K\geq C n^5 d^n r_{\min}^{\frac{1}{2}}r_{\max}^2 \cdot \scr{K}^{-1} + C nd^n \log (d) r_{\min} r_{\max}^4 \cdot \scr{K}^{-2},
    $$
    and the signal-to-noise ratio condition satisfies that
    $$
    \frac{\lambda_{\min}}{\sigma}\geq \wht{C}\frac{n d^n  r_{\max}^2 }{K} \cdot\scr{K}^{-1} + \wht{C}\left(\frac{n^3 d^n r_{\max} }{K} \right)^{\frac{1}{2}}\cdot \scr{K}^{-\frac{1}{2}},
    $$
    for some $C, \wht{C}> 0$ depending on $\kappa_0, \mu, \nu$. Then with probability at least $1- n d^{-n}$, the output of \cref{alg: online initializaton} satisfies
    \begin{equation}\label{equ: initial estimate}
        \|\bd{\mcal{T}}_0-\bd{\mcal{T}}^*\|_F\leq \scr{K}\quad \text{ and }\quad \optr{Incoh}(\bd{\mcal{T}}_0)\leq 2\kappa_0^2 \nu.
    \end{equation}
\end{proposition}
The proof of \cref{prop: intialization estimation} is provided in \cref{proof: lemma initialization estimation}. Consequently, to satisfy the initialization condition required by \cref{thm: local convergence}, the sample number $K$ required by \cref{alg: online initializaton} is $\Omega(d^n)$ if we retain only the terms that scale exponentially in $n$. In contrast, the result in \cite{cai2022provable} requires $\Omega(d^n r_{\max}^{5n/2}\kappa_0^{4n} \nu^{n})$. Since the parameters $d,r, \kappa_0, \nu$ are typically assumed to be $O(1)$, the bound established in \cref{prop: intialization estimation} shows a significant improvement over the result in \cite{cai2022provable}.

\section{Numerical Experiments}
In this section, we perform numerical experiments on QST for MPOs of $n$-qubits ($d=2$) to test the proposed oRGD algorithm. We consider the pure state case and choose three types of matrix product states. As the pure states are matrix product states, the corresponding density matrices are matrix product operators, which we use as our targets for QST. The following matrix product states are taken into consideration. (1) Randomly generated matrix product states with each core tensor's real and imaginary parts are independently sampled from the uniform distribution on $[0,1]$. We set $r_k = \min(4^k, 4^{n-k}, r), k=1, \dots, n-1$ as the TT rank of target MPO for any $r$ by default. 
(2) Greenberger-Horne-Zeilinger (GHZ) state $(|0\rangle^{\otimes n} + |1\rangle^{\otimes n})/\sqrt{2}$, whose density matrix is an MPO with rank $r_k = 4, k=1, \dots, n-1$. (3) Ground state of 1D quantum Ising model. The Hamiltonian is given by 
$$\bd{H} = -\sum_{i=1}^{n-1} \bd{S}_{i}^Z\bd{S}_{i+1}^Z + g\sum_{i=1}^n \bd{S}_i^X,$$
with $g = 1$, where $\bd{S}_i^Z$ is the Pauli-Z matrix only acting on qubit $i$. We use the DMRG implementation from the ITensor library \cite{fishman2022itensor} to obtain an approximate matrix product state representation of the ground state $|\psi\rangle$, using a maximum bond dimension $D = 16$. 

We employ two standard metrics to assess the reconstruction results $\bd{\rho}_{\text{rec}}$: the relative error of Frobenius norm $D(\bd{\rho}^*, \bd{\rho}_{\operatorname{rec}}) = \|\bd{\rho}^*- \bd{\rho}_{\operatorname{rec}}\|_F/\|\bd{\rho}^*\|_F$ and the fidelity $f(|\psi\rangle, \bd{\rho}_{\operatorname{rec}}) = |\langle\psi|\bd{\rho}_{\operatorname{rec}} |\psi \rangle|$ or $ f(|\psi\rangle, |\psi_{\operatorname{rec}}\rangle) = |\langle\psi|\psi_{\operatorname{rec}} \rangle|^2$.

\begin{figure}[H]
\hspace{10mm}
\begin{minipage}{0.4\textwidth}
    \centering
    \includegraphics[width=\linewidth]{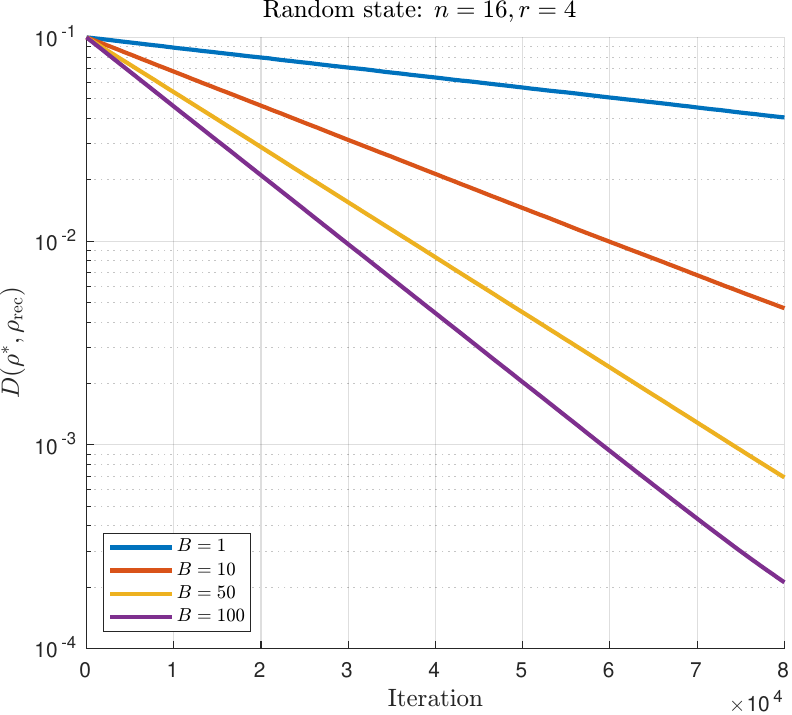}
\end{minipage}
% \hfill
\hspace{2mm}
\begin{minipage}{0.4\textwidth}
    \centering
    \includegraphics[width=\linewidth]{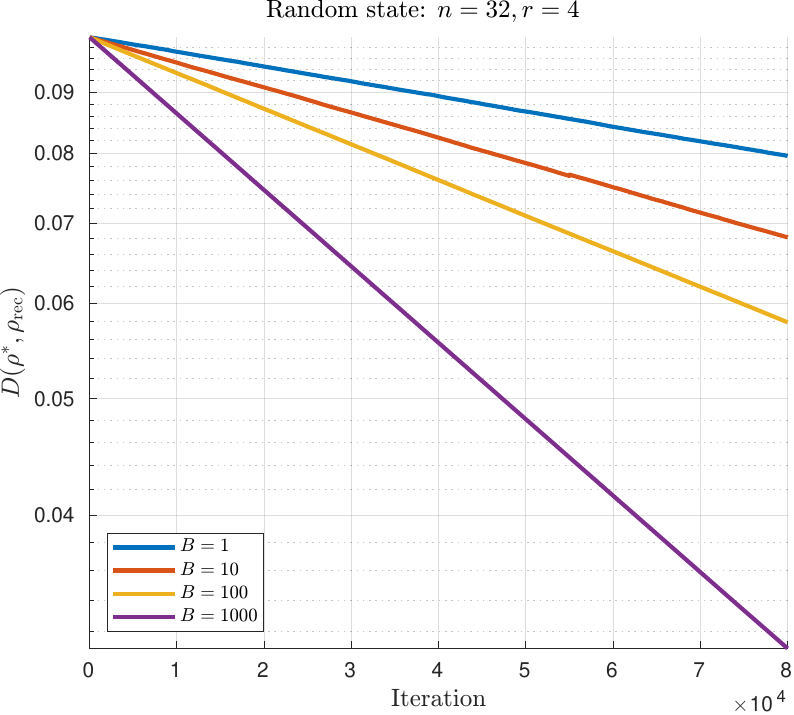}
\end{minipage}

\hspace{10mm}
\begin{minipage}{0.4\textwidth}
    \centering
    \includegraphics[width=\linewidth]{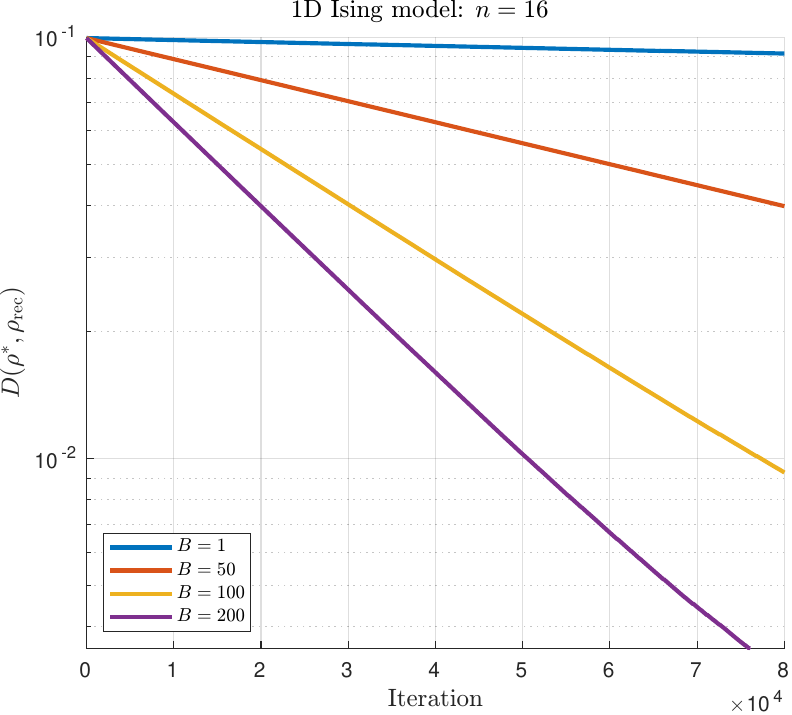}
\end{minipage}
% \hfill
\hspace{3mm}
\begin{minipage}{0.43\textwidth}
    \centering
    \includegraphics[width=\linewidth]{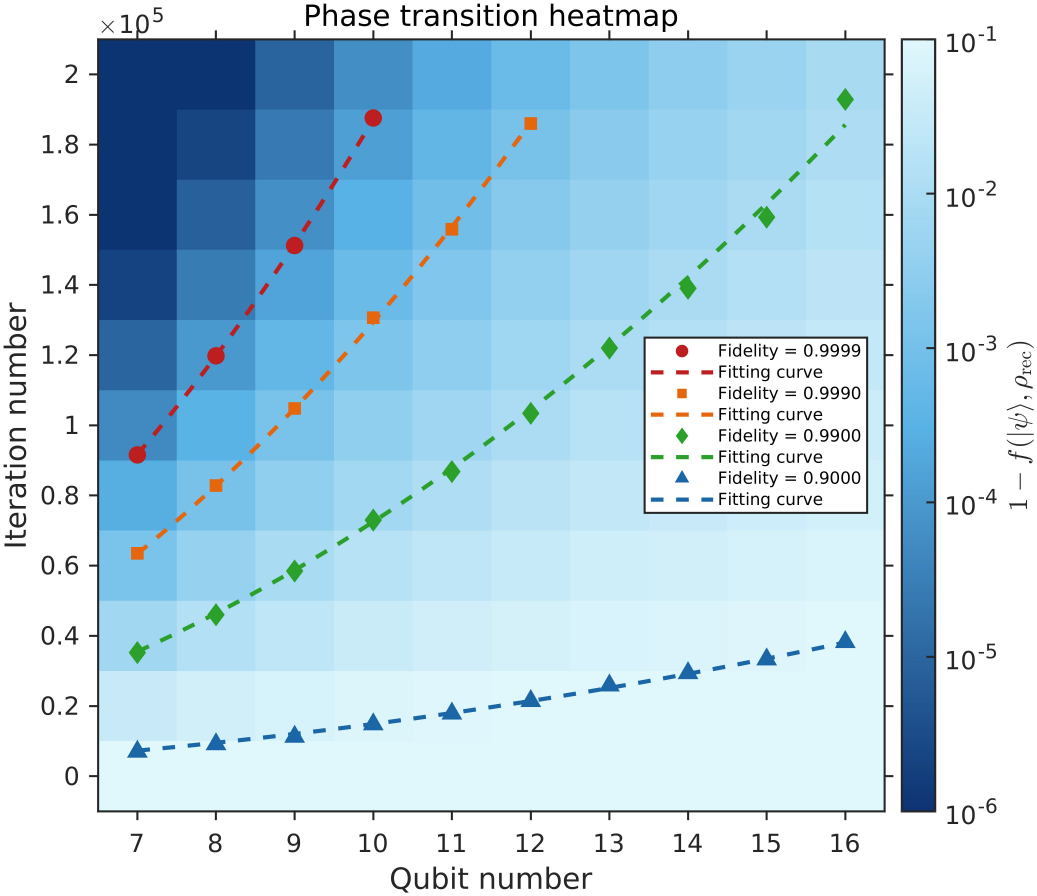}
\end{minipage}
\caption{\textbf{(Left-Top)} Random state with $n=16, r=4, \delta = 0.1$. \textbf{(Right-Top)} Random state with $n=32, r=4, \delta = 0.1$. \textbf{(Left-Bottom)} The ground state of 1D Ising model with $n=16, \delta = 0.1$. \textbf{(Right-Bottom)} The phase transition behaviour of the reconstruction fidelity as a function of $n$ (from $7$ to $16$) for a fixed rank $r=4$. The initial error $\delta = 0.7$ corresponds to a fidelity around $0.82$. Dashed lines represent quadratic fits to data points with the same fidelity.}
\label{fig: Noiseless}
\end{figure}

\textbf{Exact Pauli measurement data and local initialization.}
We first examine the noiseless setting to validate the local linear convergence of the online RGD algorithm. The initial point is constructed as $\bd{\mcal{T}}_0 = \operatorname{TTSVD}_{\bd{r}}(\bd{\mcal{T}^*} + \delta\cdot \frac{\bd{\mcal{E}}}{\|\bd{\mcal{E}}\|_F})$, where $\bd{\mcal{E}}$ is a random tensor with TT rank $\bd{r}$ and parameter $\delta$ controls the initialization error. To enhance the efficiency of the update, we perform each oRGD update using a minibatch of observations $\{\bd{Y}_{t, 1}, \dots, \bd{Y}_{t, B}\}$ sampled uniformly with replacement. The step size is set to $\eta = \frac{\alpha}{B n^2}$ where $\alpha$ is a tunable parameter. The results for quantum states with different numbers of qubits $n$ or different TT-ranks are displayed in \cref{fig: Noiseless}. In the bottom-right panel of \cref{fig: Noiseless}, we vary the number of qubits while keeping the batch size fixed at $B = 20$ and the step-size parameter $\alpha=4\times 10^{-3}$. For each setting, five independent random trials were performed and the average result is reported.

As shown in \cref{fig: Noiseless}, oRGD exhibits linear convergence for all tested quantum states, which aligns with our theoretical analysis. Furthermore, when a larger batch size $B$ is used, a larger $\alpha$ can be adopted, leading to a faster convergence rate. We also observe that the convergence becomes slower when the system size $n$ increases or when the TT‑rank is higher.
 In particular, the bottom‑right panel illustrates that, for different qubit numbers $n$, the number of iterations (and hence the total sample complexity) required to reach a fixed reconstruction error scales quadratically with $n$. This is consistent with our theoretical results.

\textbf{Noisy Pauli measurement data and random initialization.}
The initial tensor $\bd{\mcal{T}_0}$ is taken as the Pauli coefficient tensor of a randomly generated matrix product operator with rank $\bd{r}$. For each observation (Pauli basis), we perform $M = 4000$ or $8000$ measurement shots to estimate the expectation value. 

\begin{figure}[H]
\begin{center}
\begin{minipage}{0.4\textwidth}
    \centering
    \includegraphics[width=\linewidth]{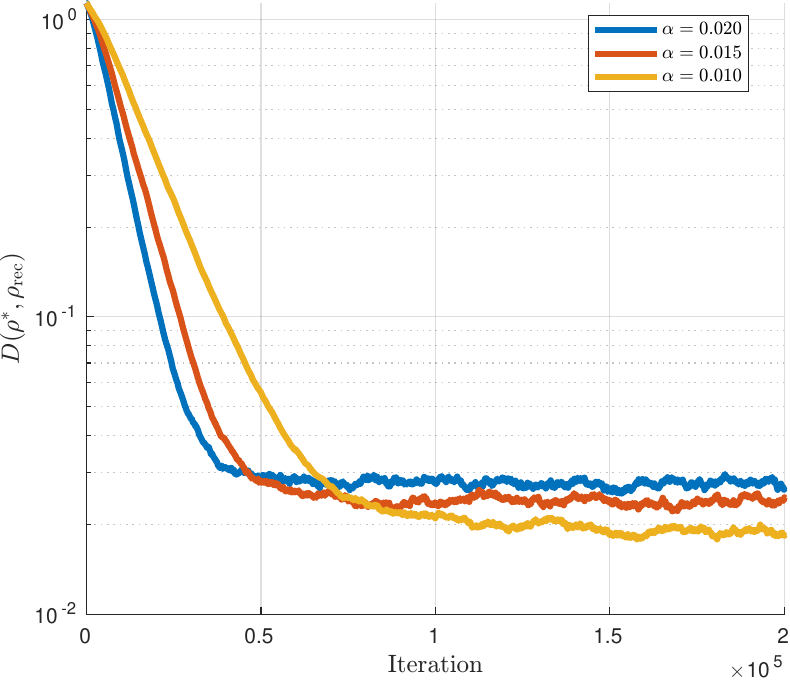}
    \vskip-2mm
\end{minipage}
% \hfill
\hspace{2mm}
\begin{minipage}{0.4\textwidth}
    \centering
    \includegraphics[width=\linewidth]{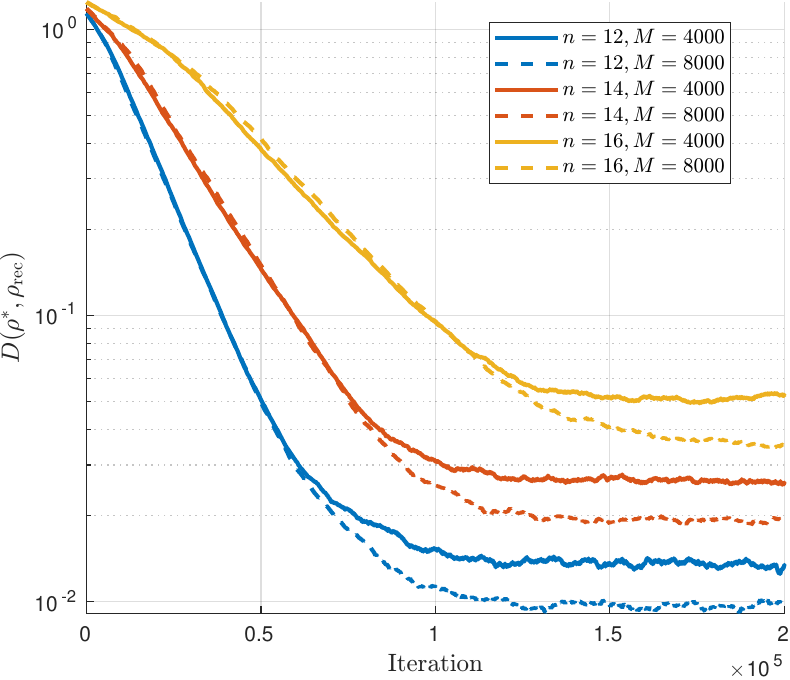}
    \vskip-2mm
\end{minipage}
\end{center}
\caption{\textbf{(Left)} Random state with $n=12, r = 4$, using noisy measurements with $M = 4000$ shots per observable. We set $B=50$ and vary the step size parameter $\alpha$ in oRGD. \textbf{(Right)} Random state with $r = 4$ and different $n$, we set $B=100, \alpha = 10^{-2}$ in oRGD.}
\label{fig: noisy result}
\end{figure}

The left figure of \cref{fig: noisy result} illustrates that online RGD initially converges at a linear rate, after which the relative error plateaus due to the measurement noise. Moreover, a smaller step size leads to slower convergence but achieves a lower final error, which is consistent with our theoretical results in \cref{thm: local convergence}. As shown in the right figure of \cref{fig: noisy result},  increasing the number of shots reduces the noise level and thus the final relative error. In general, reconstructing a quantum state with a larger qubit number requires both more samples and a higher number of measurement shots.

\textbf{Comparison of oRGD with RGD.} We compare the performance of the proposed online RGD method with the RGD algorithm \cite{cai2022provable} in the exact Pauli measurement setting. Since RGD is an offline method, we first collect the measurement data offline and then use the full dataset to implement the RGD in each iteration. In contrast, online RGD still processes a batch of samples per iteration. To meet the sample complexity requirement of RGD, the total sample budget is set to $100 \cdot 2^n$ for each qubit number $n$, and the rank is fixed at $r=4$. Both algorithms start from the same random initialization, and we report the total runtime required to reach the stopping criteria $D(\bd{\rho}^*, \bd{\rho}_{\text{rec}})\leq 10^{-3}$ in left panel of \cref{fig: RGD vs oRGD-RSGD}. The per-iteration runtime of RGD and online RGD (with batch size $B = 500$) is shown in the middle panel of \cref{fig: RGD vs oRGD-RSGD}.

As shown in \cref{fig: RGD vs oRGD-RSGD}, the total computational cost for convergence is substantially higher for RGD than oRGD, especially when $n$ is moderately large. Additionally, the per-iteration runtime of RGD scales exponentially with $n$, while oRGD scales only linearly and remains below or around $10^{-2}$ seconds across the tested range of $n$.

\begin{figure}[h]
\begin{center}
\begin{minipage}{0.32\textwidth}
    \centering
    \includegraphics[width=\linewidth]{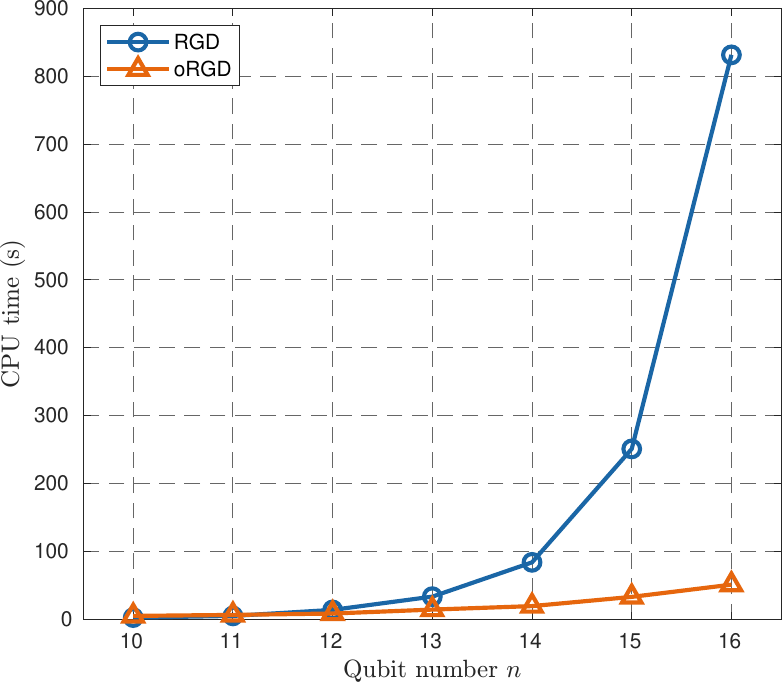}
    \vskip-2mm
\end{minipage}
% \hfill
\begin{minipage}{0.32\textwidth}
    \centering
    \includegraphics[width=\linewidth]{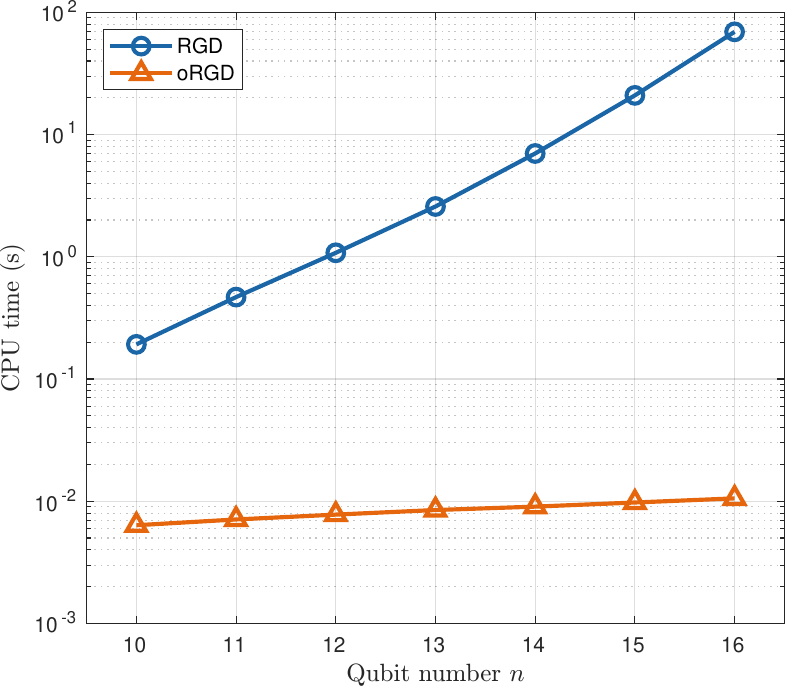}
    \vskip-2mm
    % \caption{$r=4, B=100, \alpha = 10^{-2}$.}
\end{minipage}
% \hfill
\begin{minipage}{0.32\textwidth}
    \centering
    \includegraphics[width=\linewidth]{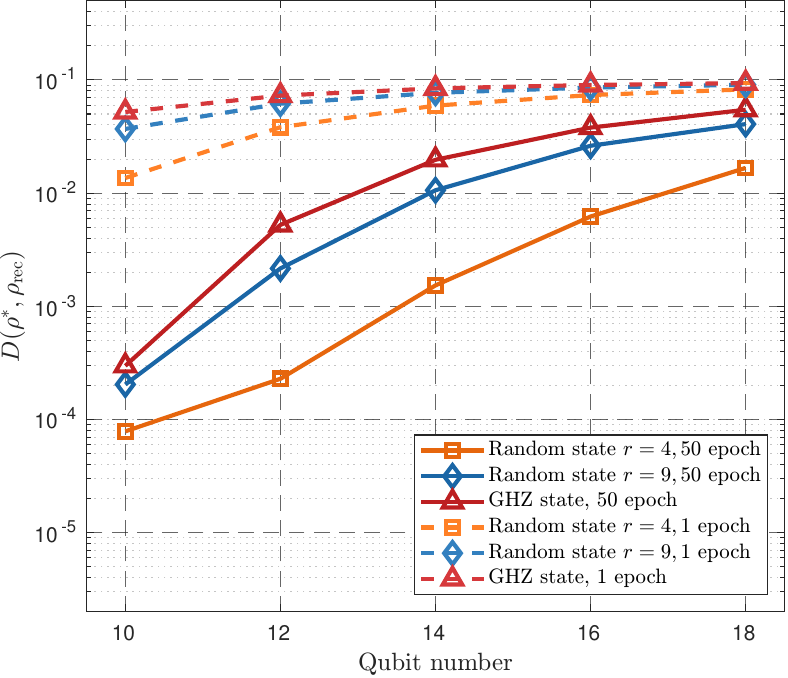}
    \vskip-2mm
    % \caption{$r=4, B=100, \alpha = 10^{-2}$.}
\end{minipage}
\end{center}
\caption{\textbf{(Left)} Total runtime required for RGD and oRGD to achieve $D(\bd{\rho}^*, \bd{\rho}_{\text{rec}})\leq 10^{-3}$. \textbf{(Middle)} Per-iteration runtime of RGD and oRGD with batch size $B = 500$. \textbf{(Right)} Reconstruction error of different quantum states using the RSGD algorithm.}
\label{fig: RGD vs oRGD-RSGD}
\end{figure}

\textbf{Riemannian stochastic gradient descent method.} We employ the RSGD algorithm with a total budget of $3\times 10^5$ exact Pauli measurement samples. The initial error $\delta = 0.1$ and the step size decays epoch-wise as $\alpha_k = \alpha_0 \times (0.9)^{k-1}$, where $k$ denotes the epoch index. We test on three target states: random states with $r=4, 9$ and the GHZ state. For these states, the batch sizes are set to $B = 3\times n^2, 5\times n^2$, $6\times n^2$ and initial step-size parameters are $\alpha_0 = 0.2, 0.2, 0.13$, respectively. The reconstruction results are shown in the right panel of \cref{fig: RGD vs oRGD-RSGD}.

The results demonstrate that reusing the measurement data can further reduce the reconstruction error. In practice, this indicates that when experimental copies of the state are constrained, reusing the measurement data can considerably enhance the reconstruction result.

\section{Technical Lemmas}

In this section, we provide several technical lemmas that will be frequently used in the following theoretical analysis. We first present the proof of the \cref{lemma:variance} below.
\begin{proof}[Proof of \cref{lemma:variance}]
Denote $\wht{\bd{A}}_s = 2^{\frac{n}{2}} \bd{A}_s$ and we first consider the estimation of the expectation $\la\wht{\bd{A}}_s , \bd{\rho} \ra$ and let $\hat{z}_s$ be the corresponding statistical error. Then $\la\wht{\bd{A}}_s, \bd{\rho} \ra = 2^{\frac{n}{2}}\la\bd{A}_s, \bd{\rho} \ra$  and $\hat{z}_s = 2^{\frac{n}{2}} z_s$.

The expectation $\langle \wht{\bd{A}}_s, \bd{\rho}\rangle$ is estimated by the $2$-outcome measurement $\{\frac{\bd{I}+\wht{\bd{A}}_s}{2}, \frac{\bd{I} - \wht{\bd{A}}_s}{2}\}$. For a single measurement, the outcome $X_s$ is a random variable with
$$
\bb{P}(X_{s} = 1) = \operatorname{Tr}\left(\frac{\bd{I}+ \wht{\bd{A}}_s}{2} \bd{\rho}\right), \quad     \bb{P}(X_{s} = -1 = \operatorname{Tr}\left(\frac{\bd{I}- \wht{\bd{A}}_s}{2} \bd{\rho}\right).
$$
Let $X_s^i, i=1, \dots, M$ be the outcomes of $M$ independent measurements. The statistical noise is defined as $\hat{z}_s = \frac{1}{M}\sum_{i=1}^M X_s^i - \langle\wht{\bd{A}}_s, \bd{\rho} \rangle$. Its expectation is
$$
\begin{aligned}
\bb{E}\left[ \hat{z}_s\right] &= \frac{1}{M}\sum_{i=1}^M\left(\bb{E}[X_s^i] - \operatorname{Tr}(\wht{\bd{A}}_s \bd{\rho}) \right)\\
&= \frac{1}{M}\sum_{i=1}^M\left( \operatorname{Tr}\left(\frac{\bd{I}+ \wht{\bd{A}}_s}{2} \bd{\rho}\right) - \operatorname{Tr}\left(\frac{\bd{I}- \wht{\bd{A}}_s}{2} \bd{\rho}\right) - \operatorname{Tr}(\wht{\bd{A}}_s \bd{\rho}) \right) \\
&=0
\end{aligned}
$$
As each $X_s^i$ takes values $\pm 1$, we have $|\hat{z}_s|\leq 2$. And the variance of $\hat{z}_s$ is
    $$
    \begin{aligned}
    \operatorname{Var}(\hat{z}_s) &= \frac{1}{M^2} \sum_{i=1}^M\operatorname{Var}(X_s^i)\\
    &\leq \frac{1}{M}\left(\operatorname{Tr}\left(\frac{\bd{I}+ \wht{\bd{A}}_s}{2} \bd{\rho}\right) +\operatorname{Tr}\left(\frac{\bd{I}- \wht{\bd{A}}_s}{2} \bd{\rho}\right) - \langle\wht{\bd{A}}_s,\bd{\rho}\rangle^2 \right)\\
    &\leq \frac{1}{M}\left(1-\langle\wht{\bd{A}}_s,\bd{\rho}\rangle^2 \right)\leq \frac{1}{M}
    \end{aligned}
    $$
    Finally, since $\hat{z}_s$ is a sum of independent bounded random variables, Hoeffding's inequality gives
    $$
    \bb{P}(|\hat{z}_s|\geq \xi )\leq 2e^{- \xi^2 M/2}, \xi\geq 0.
    $$
    Then we have 
    $$
    \optr{Var}(z_s)\leq \frac{1}{2^nM},\quad \text{and}\quad \bb{P}(|z_s|\geq \xi )\leq 2e^{- \xi^2 M2^{n-1}}, \xi\geq 0.
    $$
\end{proof}
\begin{lemma}\label{lemma: hermitian preserve l1}
For a Hermitian matrix $\bd{X}\in \mathbb{C}^{rd^{2K}\times rd^{2K}}$ satisfies that, each entry of $\bd{X}$ is the complex conjugate of another one if the indices in $\{i_t, j_t\}_{t=1}^K$ are swapped:
\begin{equation}\label{equ: swapped matrix}
        \bd{X}\left((l, \{i_t, j_t\}_{t=1}^K), (l', \{p_t, q_t\}_{t=1}^K)\right) = \overline{\bd{X}\left((l, \{j_t, i_t\}_{t=1}^K),(l', \{p_t, q_t\}_{t=1}^K)\right)}
\end{equation}
for all $l, l'\in [r]$, and $ i_t, j_t, p_t, q_t\in [d]$. Then there exists an eigenvector matrix $\bd{U}$ of $\bd{X}$ such that each eigenvector $\bd{u}\in \bb{C}^{r d^{2K}}$ in $\bd{U}$ satisfies that
\begin{equation}\label{equ: swapped eigenvector}
    \bd{u}(l, \{i_t, j_t\}_{t=1}^K) = \overline{\bd{u}(l, \{j_t, i_t\}_{t=1}^K)}
\end{equation}
for all $l\in[r]$ and $i_t, j_t\in [d]$.
\end{lemma}
\begin{proof}
    We define the following conjugate linear operator $\scr{F}: \bb{C}^{r d^{2K}}\longrightarrow \bb{C}^{rd^{2K}}:$
    $$ \scr{F}(\bd{u}) (l, \{i_t, j_t\}_{t=1}^K) = \overline{\bd{u}(l, \{j_t, i_t\}_{t=1}^K)},\quad  l\in[r], i_t, j_t\in[d].$$
    The operator $\scr{F}$ satisfies that for any $\bd{v}\in \bb{C}^{r d^K}$,  $\scr{F}^2(\bd{v}) = \bd{v}$ and 
    $\scr{F}(\bd{X} \bd{v}) = \bd{X} \scr{F}(\bd{v})$. The condition in \eqref{equ: swapped eigenvector} is equivalent to be $\scr{F}(\bd{u}) = \bd{u}$. We will prove that there exist such eigenvectors for each eigenvalue of $\bd{X}$.

    Let $\lambda\in \bb{R}$ be an eigenvalue of $\bd{X}$, $D_{\lambda}$ is the multiplicity and $\bb{E}_{\lambda}$ is the corresponding eigenvector space. For any $\bd{v}\in \bb{E}_\lambda$
    $$
    \bd{X} \scr{F}(\bd{v}) = \scr{F}(\bd{X} \bd{v}) = \scr{F}(\lambda \bd{v}) = \lambda \scr{F}\bd{v},
    $$
 thus, $\scr{F}$ maps $\bb{E}_{\lambda}$ to itself. Let $\bb{E}_{\lambda}^R := \{\bd{u}\in \bb{E}_{\lambda}: \scr{F}(\bd{u}) =\bd{u}\}$, then for any $\bd{v}\in \bb{E}_{\lambda}$, 
    $$
    \bd{v}_1 = \frac{\bd{v} +\scr{F}(\bd{v})}{2}, \quad \bd{v}_2 = \frac{\bd{v} -\scr{F}(\bd{v})}{2\sqrt{-1}}, 
    $$
    it is easy to verify that $\bd{v}_1, \bd{v}_2\in  \bb{E}_{\lambda}^{R}$ and $\bd{v} = \bd{v}_1 + \sqrt{-1} \bd{v}_2$. Thus $\bb{E}_{\lambda}^R$ forms $\bb{E}_{\lambda}$ with complex coefficients. And there exists a collection of vectors $\{\bd{w}_1, \dots, \bd{w}_{D_{\lambda}}\}$ in $\bb{E}_{\lambda}^R$ that constructs a basis of $\bb{E}_{\lambda}$. We define $\{\bd{u}_1, \dots, \bd{u}_{D_{\lambda}}\}$ as the orthogonalization and normalization of $\{\bd{w}_1, \dots, \bd{w}_{D_{\lambda}}\}$ via Gram-Schmidt process. As the inner product between any vectors in $\bb{E}_{\lambda}^R$ is a real number, each $\bd{u}_k$ can be expressed as
    $$
    \bd{u}_k = \sum_{i=1}^{D_{\lambda}} c_{k, i}\cdot \bd{w}_i,\quad  c_{k, i}\in \bb{R}.
    $$
    Thus, $\{\bd{u}_1, \dots, \bd{u}_{D_{\lambda}}\}\subset \bb{E}^R_{\lambda}$ are the eigenvectors of $\lambda$ satisfies \eqref{equ: equal hermitian condition}. By combining all those eigenvectors with different eigenvalues, we construct the target matrix $\bd{U}$. 
\end{proof}

\begin{theorem}[Martingale Concentration Inequality]\label{thm: martingale concentration}
    Suppose $X_n, n\geq 1$ is a martingale such that $X_0=0$ and $|X_i-X_{i-1}|\leq d_i, 1\leq i\leq n$ almost surely for some constant $d_i, 1\leq i\leq n$. Then for every $t\geq 0$,
    $$
    \bb{P}(|X_n|\geq t)\leq 2\exp (-\frac{t^2}{2\sum_{i=1}^n d_i^2}).
    $$
\end{theorem}

\begin{lemma}[Lemma 26 in \cite{cai2022provable}]\label{lemm: orth_comp tsp} Let $\bd{\mcal{T}}, \bd{\mcal{T}}^*\in \bb{M}_{\bd{r}}$ be two TT-rank $\bd{r}$ tensors. Suppose we have $8\|\bd{\mcal{T}} - \bd{\mcal{T}}^*\|_F\leq \lambda_{\min}$, then we have
$$
\|\mcal{P}^{\perp}_{\bb{T}}(\bd{\mcal{T}}^*)\|_F\leq \frac{12\sqrt{2} n\|\bd{\mcal{T}} - \bd{\mcal{T}}^*\|_F^2}{\lambda_{\min}},
$$
where $\bb{T}$ is the tangent space at point $\bd{\mcal{T}}$.
\end{lemma}

\begin{lemma}[Spikiness implies incoherence, Lemma 2 in \cite{cai2022provable}]\label{lemma: spiki imply incoh}
Let $\bd{\mcal{T}}\in \mathbb{M}_{\bd{r}}$ satisfies $\optr{Spiki}(\bd{\mcal{T}})\leq \nu$. Then we have 
$$
\optr{Incoh}(\bd{\mcal{T}})\leq \nu \kappa_0,
$$
where $\optr{Incoh}(\bd{\mcal{T}})$ is the incoherence parameter of $\bd{\mcal{T}}$ and $\kappa_0$ is the condition number of $\bd{\mcal{T}}$ defined by $\kappa_0 = \lambda_{\max}(\bd{\mcal{T}})/\lambda_{\min}(\bd{\mcal{T}})$.
\end{lemma}
\begin{lemma}[Incoherence implies Spikiness]\label{lemma: incoh imply spiki}
    Let $\bd{\mcal{T}}^*\in \mathbb{M}_{\bd{r}}$ satisfies $\optr{Incoh}(\bd{\mcal{T}})\leq \sqrt{\mu}$, then we have
    $$
    \optr{Spiki}(\bd{\mcal{T}})\leq \sqrt{r_{\max}} \kappa_0 \mu,
    $$
    where $r_{\max}$ is the largest rank.
\end{lemma}
\begin{proof}Notice that $\bd{\mcal{T}}^{\la i\ra} = \bd{T}^{\leq i} \bd{\Lambda}_i \bd{V}^{\top}_{i+1}$,
$$
\|\bd{\mcal{T}}\|_{\infty}\leq \lambda_{\max}(\bd{\Lambda}_{i})\cdot \|\bd{T}^{\leq i}\|_{2, \infty} \cdot \|\bd{V}_{i+1}\|_{2, \infty}\leq \lambda_{\max}(\bd{\mcal{T}}) \cdot \frac{r_i}{d^n}\cdot \mu.
$$
On the other hand, $\|\bd{\mcal{T}}\|_F\geq \sqrt{r}_i\cdot \lambda_{\min}(\bd{\mcal{T}})$. Therefore
$$
\optr{Spiki}(\bd{\mcal{T}}) = \frac{d^n \|\bd{\mcal{T}}\|_{\infty}}{\|\bd{\mcal{T}}\|_F}\leq \sqrt{r_{\max}}\kappa_0 \mu.
$$
\end{proof}

\begin{lemma}[Lemma 29 in \cite{cai2022provable}]\label{lemma: trim incoh}
    Let $\bd{\mcal{T}}^*\in \bb{M}_{\bd{r}}$ and satisfies that $\optr{Spiki}(\bd{\mcal{T}}^*)\leq \nu$. Suppose that $\bd{\mcal{W}}$ satisfies $\|\bd{\mcal{W}} - \bd{\mcal{T}}^*\|_F\leq \frac{\lambda_{\min}}{600 m \sqrt{r_{\max}} \kappa_0}$, if we choose $\xi=\frac{10 \|\bd{\mcal{W}}\|_F}{9\sqrt{d^*}}\nu$ then the incoherence of  $\optr{TTSVD}_{\bd{r}}(\optr{Trim}_{\xi}(\bd{\mcal{W}}))$ is less than $2\kappa_0^2 \nu$. Furthermore,
    $$
    \|\optr{TTSVD}_{\bd{r}}( \optr{Trim}_{\xi}(\bd{\mcal{W}})) -\bd{\mcal{T}}^* \|_F\leq \sqrt{2}\|\bd{\mcal{W}} - \bd{\mcal{T}}^*\|_F
    $$
\end{lemma}

\begin{lemma}[Lemma 26 in \cite{cai2022provable}]\label{lemma: perturbation bound of TTSVD}
Let $\bd{\mcal{T}}^*\in \bb{M}_{\bd{r}}$ and we denote $\bd{\mcal{T}} = \bd{\mcal{T}}^* +\bd{\mcal{D}}$. Then suppose $\|\bd{\mcal{D}}\|_F\lesssim_n\lambda_{\min}$ where $\lambda_{\min}$ is the smallest singular value of $\bd{\mcal{T}}^*$, we have
$$
\|\optr{TTSVD}_{\bd{r}}(\bd{\mcal{T}}) - \bd{\mcal{T}}^*\|_F^2\leq \|\bd{\mcal{D}}\|_F^2 +\frac{600 n \|\bd{\mcal{D}}\|_F^3}{\lambda_{\min}}
$$
\end{lemma}

\begin{lemma}[Theorem 2 in \cite{xia2019polynomial}]\label{lemma: matrix spectral initialization}\label{lemma: noiseless eig estimater}

Let $\bd{M}\in \bb{R}^{p_1\times p_2}$ and $\bd{X}_i = p_1 p_2 \mcal{P}_{\omega_i}\bd{M}, \bd{Y}_j = p_1 p_2 \mcal{P}_{\omega_j'}\bd{M}$, where $\omega_i \in \Omega_1, \omega_j'\in \Omega_2$ are independently and uniformly sampled from $[p_1]\times [p_2]$ and $|\Omega_1| = |\Omega_2| = K$. Denote $\bd{N} = \bd{M} \bd{M}^\top$ and $\widetilde{\bd{N}} = \frac{1}{2K^2}\sum_{i, j}(\bd{X}_i \bd{Y}_j^\top + \bd{Y}_j \bd{X}_i^\top)$, then with high probability exceeding $1-p^{-\alpha}$ with $p =\max \{p_1, p_2\}$, we have
$$
\|\widetilde{\bd{N}} - \bd{N}\|\leq C \alpha^2 \frac{p_1^{3/2}p_2^{3/2}\log (p)}{K}\left[\left(1 + \frac{p_1}{p_2}\right)^{1/2} + \frac{p_1^{1/2}p_2^{1/2}}{K} + \left(\frac{K}{p_2 \log (p)}\right)^{1/2} \right]\cdot \|\bd{M}\|_{\infty}^2.
$$
\end{lemma}

\begin{lemma}[Lemma 18 in \cite{cai2022provable}]\label{lemma: noisy truncated eig estimator}
Let $\bd{M} \in \mathbb{R}^{p_1 \times p_2}$ and $\bd{X}_i=p_1 p_2\left(\bd{M}_{\omega_i}+\xi_i\right) \bd{E}_{\omega_i}, \bd{Y}_j=p_1 p_2\left(\bd{M}_{\omega_j^{\prime}}+\xi_j^{\prime}\right) \bd{E}_{\omega_j^{\prime}}$, where $\omega_i \in \Omega_1, \omega_j^{\prime} \in \Omega_2$ are independently and uniformly sampled from $\left[p_1\right] \times\left[p_2\right]$ and $\left|\Omega_1\right|=\left|\Omega_2\right|= K$ with $K \leq p_1 p_2$, and $\left\{\bd{E}_\omega\right\}_{\omega \in\left[p_1\right] \times\left[p_2\right]}$ is the standard basis for $\mathbb{R}^{p_1 \times p_2}$ and $\xi,\left\{\xi_i\right\}_{i=1}^K,\left\{\xi_j^{\prime}\right\}_{j=1}^K$ are i.i.d. $\hat{\sigma}$ subgaussian random variables with variance Var $\xi^2 \leq C_1 \hat{\sigma}^2$ for some absolute constant $C_1>0$. Let $\bd{U} \in \mathbb{R}^{p_1 \times r}$ be the orthogonal matrix such that Incoh $(\bd{U}) \leq \sqrt{\mu}$. Then for any $\alpha \geq 1$, with probability exceeding $1-9 p^{-\alpha}$ where $p=\max \left\{p_1, p_2\right\}$, we have for some absolute constant $C_2>0$,

$$
\begin{aligned}
& \left\|\frac{1}{2 K^2} \sum_{i, j}\left(\bd{U}^T \bd{X}_i \bd{Y}_j^T \bd{U}+\bd{U}^T \bd{Y}_j \bd{X}_i^T \bd{U}\right)-\bd{U}^T \bd{M} \bd{M}^T \bd{U}\right\| \\
\leq & C_\alpha\|\bd{M}\|_{\infty} \hat{\sigma}\left(\mu r \frac{p_1 p_2}{K} \log (p)+\mu r \frac{p_1 p_2^2}{K^2} \log ^2(p)+\frac{p_1 p_2}{\sqrt{K}} \sqrt{\mu r \log (p)}\right) \\
+ & C_\alpha \hat{\sigma}^2 \frac{p_1 p_2}{K}\left(\mu r \log ^{3 / 2}(p)+\mu r \log ^{5 / 2}(p) \frac{\sqrt{p_2} \sqrt{p_2 \vee r}}{K}\right) \\
+ & C_\alpha \log ^2(p) \frac{p_1 p_2\|\bd{M}\|_{\infty}^2}{K}\left(\mu r p_2^{1 / 2}+\frac{\mu r p_2}{K}+\left(\frac{\mu r K}{\log ^3(p)}\right)^{1 / 2}\right) .
\end{aligned}
$$
Let $\bd{W} = \frac{p_1p_2}{K}\sum_{i=1}^K\xi_i \bd{U}^\top \bd{E}_{\omega_i}$ and $\bd{W}^\prime = \frac{p_1p_2}{K}\sum_{j=1}^K\xi_j^\prime \bd{U}^\top \bd{E}_{\omega_j^\prime}$, then with probability exceeding $1-2p^{-\alpha}$
$$
\max\{\|\bd{W}\|, \|\bd{W}^\prime\|\}\leq C \alpha \left(\sqrt{\frac{p_1p_2 \max(p_2, r)}{K} \log (p)} + \frac{p_1 p_2}{K}\mu r\log (p)\right)\hat{\sigma}
$$
\end{lemma}
\begin{lemma}[Lemma 36 in \cite{cai2022provable}]\label{lemma: last truncation}
Let $\bd{M}\in \bb{R}^{p_1\times p_2}$ and $\bd{X}_i = p_1p_2\mcal{P}_{\omega_i} (\bd{M})$, where $\omega_i\in \Omega$ is independently and uniformly sampled in $[p_1]\times [p_2]$ and $|\Omega| =  K$. Let $\bd{U}\in \bb{R}^{p_1\times r}$ be the orthogonal matrix such that $\optr{Incoh}(\bd{U})\leq \sqrt{\mu}$. Then with probability exceeding $1-p^{-\alpha}$, we have
$$
\|\bd{U}^\top (\frac{p_1p_2}{K}\mcal{P}_{\Omega} (\bd{M}) - \bd{M} )\|\leq C\alpha \left(\frac{\sqrt{p_1}p_2\sqrt{\mu r}\|\bd{M}\|_{\infty}\log (p)}{K} + \sqrt{\frac{p_1p_2\|\bd{M}\|_{\infty}^2\max(\mu r, p_2)\log(p)}{K}} \right)
$$    
\end{lemma}

\begin{lemma}[Remark 6.2 in \cite{keshavan2010matrix}] \label{lemma: Incoherence of matrix}

Let $\bd{U}, \bd{X} \in \bb{R}^{p\times r}$ be orthogonal and $\optr{Incoh}(\bd{U})\leq \sqrt{\mu}$ and $d_p(\bd{U}, \bd{X})\leq \delta \leq \frac{1}{16\pi}$. Then $\wht{\bd{X}}$ satisfies $\optr{Incoh}(\wht{\bd{X}})\leq \sqrt{3\mu}$ and $d_p(\wht{\bd{X}}, \bd{U})\leq 4\pi \delta$, where 
$$
\bar{\bd{X}}^i = \frac{\bd{X}^i}{\|\bd{X}^i\|_{2}}\cdot \min \{\|\bd{X}^i\|_2, \sqrt{\frac{\mu r}{p}}\}, \ \wht{\bd{X}} = \bar{\bd{X}}(\bar{\bd{X}}^\top \bar{\bd{X}})^{-1/2}.
$$
    
\end{lemma}

\section*{Proof of \cref{thm: equal Herm}}\label{proof: equal Herm}
% For a density matrix, one can apply the TT-SVD algorithm to obtain the MPO approximation.
If there exists a decomposition $\bd{\rho}=[\bd{\mcal{U}}_1, \bd{\mcal{U}}_2, \dots, \bd{\mcal{U}}_n]$ satisfying \eqref{equ: equal hermitian condition}, then $\bd{\rho}$ is a Hermitian matrix due to the definition in \eqref{equ: MPO}.
On the other hand, the decomposition in \eqref{equ: MPO} is not unique, we introduce the following algorithm to construct the decomposition satisfying \eqref{equ: equal hermitian condition} for the Hermitian matrix $\bd{\rho}$ with bond dimension $\bd{r}$.

\begin{center}    
\begin{minipage}{0.75\linewidth}
\begin{algorithm}[H]
\caption{Hermitian Matrix Product Operator Decomposition}
\begin{algorithmic}\label{alg: HermTTSVD}
\REQUIRE A Hermitian matrix $\bd{\rho}\in \mathbb{C}^{d^n\times d^n}$ and target bond dimension $r_0 = r_n = 1, \bd{r} = \left(r_1, r_2, \cdots, r_{n-1}\right)$.\\
\STATE Obtain the reshaped $n$-order tensor $\bd{\mcal{Y}}\in \bb{C}^{d^2\times \cdots \times d^2}$ satisfies:
$$
\bd{\mcal{Y}}\left((i_1, j_1), \dots, (i_n, j_n)\right) = \bd{\rho}\left((i_1, \dots, i_n), (j_1, \dots, j_n)\right).
$$
% \STATE Set $\mcal{A} = \mcal{Y}$
\FOR{$k = 1,\cdots, n-1$}
\STATE Set $\bd{M}_k = \operatorname{Reshape}\left(\bd{\mcal{Y}}, [r_{k-1} d^2, d^{2n-2k}]\right)$.
\STATE $\bd{V}_k \longleftarrow$ top left-$r_k$ eigenvectors of $\bd{M}_k \bd{M}_k^\dagger$.
\STATE Compute $\bd{\wht{V}}_k$:
$$\bd{\wht{V}}_k\left((l_{k-1}, i_k, j_k), l_{k}\right) = \overline{\bd{V}_k\left((l_{k-1}, j_k, i_k), l_k\right)},\ l_k\in [r_k], i_k, j_k\in [d].$$
\STATE Set $\bd{U}_{k,1} = (\bd{V}_k + \bd{\wht{V}}_k)/2,\ \bd{U}_{k,2} = (\bd{V}_k - \bd{\wht{V}}_k)/(2 \sqrt{-1})$.
\STATE $\bd{U}_k\longleftarrow$ $r_k$ orthogonal basis of $[\bd{U}_{k, 1}, \bd{U}_{k,2}]$ via Gram-Schmidt Process. 
% \STATE Perform truncated-$r_i$ SVD: $[U_i, \Sigma_i, V_i] \longleftarrow \operatorname{SVD}_{r_i}(M_i)$.
\STATE Set $\bd{\mcal{U}}_k = \operatorname{Reshape}\left(\bd{U}_k, [r_{k-1}, d, d,  r_{k}]\right)$.
\STATE Set $\bd{\mcal{Y}} = \operatorname{Reshape}\left(\bd{U}_k^\dagger \bd{M}_k, [r_{k}, d^2, \dots, d^2]\right)$ as a $n-k+1$-order tensor.
\ENDFOR
\STATE $\bd{\mcal{U}}_n = \operatorname{Reshape}\left(\bd{\mcal{Y}}, [r_{n-1}, d, d]\right)$
\ENSURE  $\bd{\mcal{U}}_1, \bd{\mcal{U}}_2,\ldots, \bd{\mcal{U}}_n$ satisfying \eqref{equ: equal hermitian condition}.
\end{algorithmic} 
% \label{alg: TT-SVD}
\end{algorithm}
\end{minipage}
\end{center}
We will prove that each $\bd{U}_k$ in \cref{alg: HermTTSVD} is an eigenvector matrix of $\bd{M}_k \bd{M}_k^\dagger$ and satisfies \eqref{equ: equal hermitian condition}, thus the \cref{alg: HermTTSVD} is a special Tensor-Train SVD algorithm combined with specific singular vector selection. Then, if we set the input $\bd{r}$ of \cref{alg: HermTTSVD} is the bond dimension of $\bd{\rho}$,  $\bd{\rho} = [\bd{\mcal{U}}_1, \bd{\mcal{U}}_2,\ldots, \bd{\mcal{U}}_n]$ is the target Hermitian decomposition.

By the definition of $\bd{\mcal{Y}}$, $\bd{M}_1 \in \bb{C}^{d^2\times d^{2n -2}}$ satisfies:
$$
\bd{M}_1\left((i_{1}, j_{1}), (\{i_t, j_t\}_{t=2}^n)\right) = \overline{\bd{M}_1\left((j_{1}, i_{1}), (\{j_t, i_t\}_{t=2}^n)\right)}, \quad i_t, j_t\in [d]
$$

For $1\leq k\leq n-1$. We assume that $\bd{M}_k \in \bb{C}^{r_{k-1} d^2\times d^{2n - 2k}}$ satisfies that
$$
\bd{M}_k\left((l_{k-1},i_{k}, j_{k}), (\{i_t, j_t\}_{t=k+1}^n)\right) = \overline{\bd{M}_k\left((l_{k-1},j_{k}, i_{k}), (\{j_t, i_t\}_{t=k+1}^n)\right)}, \quad i_t, j_t\in [d], l_{k-1}\in [r_{k-1}]
$$
Then, it is straightforward to verify that $\bd{M}_k \bd{M}_k^\dagger\in \bb{C}^{r_{k-1}d^2\times r_{k-1}d^2}$ satisfies
$$
\begin{aligned}
&\bd{M}_{k}\bd{M}_{k}^\dagger\left((l_{k-1},i_{k}, j_{k}), (l_{k-1}^{\prime},p_{k}, q_{k})\right)\\
&=  \sum_{\{i_t, j_t\}_{t=k+1}^n}\bd{M}_{k}\left((l_{k-1},i_{k}, j_{k}), (\{i_t, j_t\}_{t=k+1}^n)\right) \overline{\bd{M}_{k}\left((l_{k-1}',p_{k}, q_{k}), (\{i_t, j_t\}_{t=k+1}^n)\right)}\\
&=  \sum_{\{i_t, j_t\}_{t=k+1}^n}\bd{M}_{k}\left((l_{k-1}',q_{k}, p_{k}), (\{j_t, i_t\}_{t=k+1}^n)\right) \overline{\bd{M}_{k}\left((l_{k-1},j_{k}, i_{k}), (\{j_t, i_t\}_{t=k+1}^n)\right)}\\
&= \bd{M}_{k}\bd{M}_{k}^\dagger\left((l_{k-1}^{\prime},q_{k}, p_{k}), (l_{k-1},j_{k}, i_{k})\right) = \overline{\bd{M}_{k}\bd{M}_{k}^\dagger\left((l_{k-1},j_{k}, i_{k}), (l_{k-1}^{\prime},q_{k}, p_{k})\right)}\\
\end{aligned}
$$
Thus, by \cref{lemma: hermitian preserve l1}, $\bd{U}_k$ is an eigenvector matrix of $\bd{M}_k \bd{M}_k^{\dagger}$ and 
$$
\bd{\mcal{U}}_k(l_{k-1}, i_k, j_k, l_k) = \bd{U}_k\left((l_{k-1}, i_k, j_k), l_k\right) = \overline{\bd{U}_k\left((l_{k-1}, j_k, i_k), l_k\right)} = \overline{\bd{\mcal{U}}_k(l_{k-1}, j_k, i_k, l_k)}
$$
for all $i_k, j_k\in [d]$. Thus, by definition $\bd{M}_{k+1}\in \bb{C}^{r_k d^2\times d^{2n-2k-2}}$ satisfies:
$$
\begin{aligned}
\bd{M}_{k+1}\left((l_{k},i_{k+1}, j_{k+1}), (\{i_t, j_t\}_{t=k+2}^n)\right) &= \sum_{i_{k}, j_k, l_{k-1}} \overline{\bd{U}_{k}\left( (l_{k-1}, i_k, j_k), l_k\right)} \bd{M}_{k}\left((l_{k-1},i_{k}, j_{k}), (\{i_t, j_t\}_{t=k+1}^n)\right)\\
&= \sum_{i_k, j_k, l_{k-1}} \bd{U}_{k}\left( (l_{k-1}, j_k, i_k), l_k\right)\overline{\bd{M}_{k}\left((l_{k-1},j_{k}, i_{k}), (\{j_t, i_t\}_{t=k+1}^n)\right)}\\
&= \overline{\bd{M}_{k+1}\left((l_{k},j_{k+1}, i_{k+1}), (\{j_t, i_t\}_{t=k+2}^n)\right)}
\end{aligned}
$$
Recursively, the core tensors $\bd{\mcal{U}}_1, \cdots, \bd{\mcal{U}}_{n-1}$ satisfies \eqref{equ: equal hermitian condition} by \cref{lemma: hermitian preserve l1}. For core tensor $\bd{\mcal{U}}_n$ we have
$$
\bd{\mcal{U}}_n\left(l_{n-1}, i_n, j_n\right) = \bd{M}_n(l_{n-1},(i_{n}, j_{n})) = \overline{\bd{M}_n(l_{n-1}, (j_{n}, i_{n}))} = \overline{\bd{\mcal{U}}_n\left(l_{n-1}, j_n, i_n\right)}
$$
Thus, we have completed the proof.

\section*{Proof of \cref{thm: local convergence}}\label{proof: local convergence}
We first define the following event:
$$
\begin{aligned}
\mcal{E}_t = \Big\{\forall\ 0\leq l\leq t, &\|\bd{\mcal{T}}_l - 
\bd{\mcal{T}}^*\|_F^2\leq 2(1-\frac{1}{4} \eta)^l \|\bd{\mcal{T}}_0 - \bd{\mcal{T}}^*\|_F^2 + 10 \eta \cdot\overline{\optr{dof}} \sigma^2,\\
&|\epsilon_l|\lesssim \sigma\sqrt{\log d}, \optr{Incoh}(\bd{\mcal{T}}_l)\leq \sqrt{\mu_0} \Big\}
\end{aligned}
$$
where $\mu_0 = 4 \kappa_0^2 \mu$. It is easy to verify that the $\bb{P}(\mcal{E}_0) = \bb{P}(|\epsilon_0|\lesssim \sigma\sqrt{\log d})\geq 1-d^{-20}$. Also, we have the initialization requirement as $\|\bd{\mcal{T}}_0 - \bd{\mcal{T}}^*\|_F\lesssim n^{-1} \eta \cdot \lambda_{\min}.$ Here we recall that $\lambda_{\min} = \sigma_{\min}(\bd{\mcal{T}}^*)$, $\lambda_{\max} = \sigma_{\max}(\bd{\mcal{T}}^*)$ and $\kappa_0 = \lambda_{\max}/\lambda_{\min}$. We first analyse the singular value of $\bd{\mcal{T}}_t$. From $\|\bd{\mcal{T}}_t - \bd{\mcal{T}}\|_F^2\leq 2 \|\bd{\mcal{T}}_0 - \bd{\mcal{T}}^*\|_F^2 + 10 \eta\cdot \overline{\optr{dof}} \sigma^2$. We have 
$$
\lambda_{\max}(\bd{\mcal{T}}_t)\leq \lambda_{\max} + \|\bd{\mcal{T}}_t - \bd{\mcal{T}}^*\|_F\leq \frac{3}{2} \lambda_{\max},\quad \lambda_{\min}(\bd{\mcal{T}}_t)\geq \lambda_{\min} - \|\bd{\mcal{T}}_t - \bd{\mcal{T}}^*\|_F\geq \frac{1}{2}\lambda_{\min}
$$
as long as ${\eta \cdot \overline{\optr{dof}} \sigma^2 \lesssim \lambda_{\min}^2}$. Then we estimate the spikiness of $\bd{\mcal{T}}_t$. By \cref{lemma: incoh imply spiki} we have $\optr{Spiki}(\bd{\mcal{T}}_t)\leq \sqrt{r_{\max}} \kappa_0 \mu_0 =: \nu_0$. We now derive the following bounds, which will be used frequently.

\textit{Estimation of $\|\mcal{P}_{\bb{T}_t}\bd{\mcal{E}}_{\omega}\|_F$ and $\|\mcal{P}_{\bb{T}_t}\bd{\mcal{X}}_{t}\|_F$.} We write $\bd{\mcal{T}}_t = [\bd{T}_1, \bd{T}_2, \dots, \bd{T}_n]$ where $\bd{\mcal{T}}_t$ is $\mu_0$ incoherent under $\mcal{E}_t$. Then for any $\omega\in [d^2]\times [d^2]\times \cdots\times [d^2]$,
$$\mcal{P}_{\bb{T}_t}\bd{\mcal{E}}_{\omega} = \sum_{i=1}^n \delta \bd{\mcal{E}}_{\omega}^i\text{,  where }\delta \bd{\mcal{E}}_{\omega}^{i} = [\bd{T}_1, \dots, \bd{T}_{i-1}, \bd{E}_{\omega}^i, \bd{T}_{i+1}, \dots, \bd{T}_n].$$
And each $\bd{{E}}_{\omega}^i$ has the same formulation as in \eqref{equ: tgs projection}. For $i=2,\dots, n-1$, using the incoherence of $\bd{\mcal{T}}_t$ we have
$$
\begin{aligned}
    \|\delta \bd{\mcal{E}}_{\omega}^i\|_F &= \|(\bd{T}^{\leq i-1}\otimes \bd{{I}}_{d_i})L(\bd{E}_{\omega}^i)\bd{T}^{\geq i+1}\|_F\\
    &= \|(\bd{T}^{\leq i-1}\otimes \bd{I}_{d_i}) (\bd{I}_{d_i r_i} - L(\bd{T}_i)L(\bd{T}_i)^\top)(\bd{T}^{\leq i-1}\otimes \bd{I}_{d_i})^\top \bd{\mcal{E}}_{\omega}^{\la i\ra} (\bd{T}^{\geq i+1})^\top (\bd{T}^{\geq i+1}(\bd{T}^{\geq i+1})^\top)^{-1} \bd{T}^{\geq i+1}\|_F\\
    &\leq \|(\bd{T}^{\leq i-1}\otimes \bd{I}_{d_i})^\top \bd{\mcal{E}}_{\omega}^{\la i\ra} \bd{V}_{i+1} \|_F\\
    &\leq \|\bd{T}^{\leq i-1}\|_{2, \infty}\cdot \|\bd{V}_{i+1}\|_{2, \infty}\\
    &\leq \mu_0 \frac{\sqrt{r_{i-1}r_i}}{d^{n-1}}.
\end{aligned}
$$
For $i=1$ or $i=n$.
$$
\begin{aligned}
\|\delta \bd{\mcal{E}}_{\omega}^n\|_F = \|(\bd{T}^{\leq n-1}\otimes \bd{I}_{d_n})^\top \bd{\mcal{E}}_{\omega}^{\la n\ra}\|_F\leq \|\bd{T}^{\leq n-1}\|_{2, \infty}\leq \sqrt{\mu_0}\frac{\sqrt{r_{n-1}}}{d^{n-1}},\quad
\|\delta \bd{\mcal{E}}_{\omega}^1\|_F\leq \sqrt{\mu_0}\frac{\sqrt{r_{1}}}{d^{n-1}}.
\end{aligned}
$$
Thus,
$$
\|\mcal{P}_{\bb{T}_t} \bd{\mcal{E}}_{\omega}\|_F^2 = \sum_{i=1}^n \|\delta \bd{\mcal{E}}_{\omega}^i\|_F^2\leq \overline{\optr{dof}}\cdot \frac{\mu_0^2}{d^{2n}}.
$$
Meanwhile $\|\mcal{P}_{\bb{T}_t}\bd{\mcal{X}}_t\|_F^2\leq \overline{\optr{dof}} \mu_0^2$. And
\begin{equation}\label{equ2: E_t P_{T_t} X_t}
\begin{aligned}
\bb{E}_t\|\mcal{P}_{\bb{T}_t}\bd{\mcal{X}}_t\|_F^2 &= \sum_{\omega}\sum_{i=1}^n\|\delta \bd{\mcal{E}}_{\omega}^i\|_F^2\\
&\leq \sum_{\omega} \left( \sum_{i=1}^{n-1} \|(\bd{T}^{\leq i-1}\otimes \bd{I}_{d_i})\bd{\mcal{E}}_{\omega}^{\la i\ra} \bd{V}_{i+1}\|_F^2 + \|(\bd{T}^{\leq n-1}\otimes \bd{I}_{d_n})^\top \bd{\mcal{E}}_{\omega}^{\la n\ra}\|_F^2 \right)\\
&= r_1 d^2 + \sum_{i=1}^{n-2}r_i d^2 r_{i+1} + d^2 r_{n-1} = \overline{\optr{dof}}.
\end{aligned}
\end{equation}
\textit{Estimation of $\|\bd{\mcal{T}}_t - \bd{\mcal{T}}^*\|_{\infty}$.}
We first define the projection operator $\mcal{P}_{\bd{T}^{\leq i}}$ as:
$$
(\mcal{P}_{\bd{T}^{\leq i}} (\bd{\mcal{Z}}))^{\la i \ra} = (\bd{T}^{\leq i} \bd{T}^{\leq i \top}) (\bd{\mcal{Z}}^{\la i\ra})
$$
Then we have
$$
\bd{\mcal{T}}_t - \bd{\mcal{T}}^* = \mcal{P}_{\bd{T}^{*\leq n-1}} (\bd{\mcal{T}}_t - \bd{\mcal{T}}^*) + \sum_{i=1}^{n-1} (\mcal{P}_{\bd{T}^{*\leq i-1}} - \mcal{P}_{\bd{T}^{*\leq i}}) (\bd{\mcal{T}}_t - \bd{\mcal{T}}^*)
$$
where
$$
\begin{aligned}
\left((\mcal{P}_{\bd{T}^{*\leq i-1}} - \mcal{P}_{\bd{T}^{*\leq i}}) (\bd{\mcal{T}}_t - \bd{\mcal{T}}^*)\right)^{\la i\ra} &= \left((\mcal{P}_{\bd{T}^{*\leq i-1}} - \mcal{P}_{\bd{T}^{*\leq i}}) (\bd{\mcal{T}}_t)\right)^{\la i\ra} \bd{V}_{i+1} \bd{V}_{i+1}^\top\\
&= \left((\mcal{P}_{\bd{T}^{*\leq i-1}} - \mcal{P}_{\bd{T}^{*\leq i}}) (\bd{\mcal{T}}_t - \bd{\mcal{T}}^*)\right)^{\la i\ra} \bd{V}_{i+1} \bd{V}_{i+1}^\top
\end{aligned}
$$
here the first equality holds since $(\mcal{P}_{\bd{T}^{*\leq i-1}} - \mcal{P}_{\bd{T}^{*\leq i}}) (\bd{\mcal{T}}^*) = 0$ and $(\bd{\mcal{T}}_t)^{\la i\ra} \bd{V}_{i+1} \bd{V}_{i+1}^\top = (\bd{\mcal{T}}_t)^{\la i\ra}$.

Thus, we have:
$$
\begin{aligned}
    |\la \mcal{P}_{\bd{T}^{*\leq n-1}}(\bd{\mcal{T}}_t - \bd{\mcal{T}}^*), \bd{\mcal{E}}_{\omega} \ra| &= |\la \bd{\mcal{T}}_t - \bd{\mcal{T}}^*, \mcal{P}_{\bd{T}^{*\leq n-1}}(\bd{\mcal{E}}_{\omega}) \ra|\\
    &\leq \|\bd{\mcal{T}}_t - \bd{\mcal{T}}^*\|_F \cdot \|\bd{T}^{*\leq n-1}\|_{2, \infty}\\
    &\leq \|\bd{\mcal{T}}_t - \bd{\mcal{T}}^*\|_F \cdot \frac{\sqrt{\mu_0 r_{n-1}}}{d^{n-1}}.
\end{aligned}
$$
and 
$$
\begin{aligned}
&|\la (\mcal{P}_{\bd{T}^{*\leq i-1}} - \mcal{P}_{\bd{T}^{*\leq i}})(\bd{\mcal{T}}_t - \bd{\mcal{T}}^*), \bd{\mcal{E}}_{\omega} \ra| = |\la \bd{\mcal{T}}_t - \bd{\mcal{T}}^*, \left((\mcal{P}_{\bd{T}^{*\leq i-1}} - \mcal{P}_{\bd{T}^{*\leq i}}) (\bd{\mcal{E}}_{\omega})\right)^{\la i\ra} \bd{V}_{i+1} \bd{V}_{i+1}^\top \ra|\\
&\qquad\leq \|\bd{\mcal{T}}_t - \bd{\mcal{T}}^*\|_F\cdot \|(\bd{T}^{*\leq i-1}\otimes \bd{I}_{d_i})(\bd{I}-L(\bd{T}^*_i)L(\bd{T}^*_i)^\top) (\bd{T}^{*\leq i-1}\otimes \bd{I}_{d_i})^\top \bd{\mcal{E}}_{\omega}^{\la i\ra} \bd{V}_{i+1} \bd{V}_{i+1}^\top\|_F\\
&\qquad\leq \|\bd{\mcal{T}}_t - \bd{\mcal{T}}^*\|_F \cdot\|\bd{T}^{*\leq i-1}\|_{2, \infty}\cdot \|\bd{V}_{i+1}\|_{2, \infty}\\
&\qquad\leq \left\{\begin{array}{ll}
    \|\bd{\mcal{T}}_t - \bd{\mcal{T}}^*\|_F \cdot\mu_0\frac{\sqrt{r_{i-1}r_i}}{d^{n-1}}, & i=2, \dots, n-1 \\
    \|\bd{\mcal{T}}_t - \bd{\mcal{T}}^*\|_F\cdot\frac{\sqrt{\mu_0 r_1}}{d^{n-1}}, & i=1
\end{array}\right.
\end{aligned}
$$
Thus
\begin{equation}
\label{equ2: l_infty norm T_t-T^*}
\begin{aligned}
\|\bd{\mcal{T}}_t - \bd{\mcal{T}}^*\|_{\infty}&\leq \|\bd{\mcal{T}}_t - \bd{\mcal{T}}^*\|_F \cdot \frac{\sqrt{\mu_0}}{d^{n-1}}(\sqrt{r_1} + \sqrt{r_n} + \sum_{i=2}^{n-1} \sqrt{\mu_0 r_{i-1} r_i})\\
&\leq\|\bd{\mcal{T}}_t - \bd{\mcal{T}}^*\|_F\cdot  \frac{n\mu_0 r_{\max}}{d^{n-1}}.
\end{aligned}
\end{equation}

\textit{Estimation of $\|\mcal{P}_{\bb{T}_t}\bd{\mcal{G}}_t\|_F$.} Recall that $\bd{\mcal{G}}_t = (\la\bd{\mcal{X}}_t, \bd{\mcal{T}}_t - \bd{\mcal{T}}^* \ra - \epsilon_t) \bd{\mcal{X}}_t$.
% We assume $\optr{Spiki}(\mcal{T}^*):=\frac{\sqrt{d^*}\|\mcal{T}^*\|_{\infty}}{\|\mcal{T}^*\|_F}\leq \nu_0$.

\begin{equation}\label{equ2: P_T_t G_t}
\begin{aligned}
\|\mcal{P}_{\bb{T}_t}\bd{\mcal{G}}_t\|_F^2&\leq 2(\la\bd{\mcal{X}}_t, \bd{\mcal{T}}_t-\bd{\mcal{T}}^* \ra^2 + |\epsilon_t|^2)\|\mcal{P}_{\bb{T}_t} \bd{\mcal{X}}_t\|_F^2\\
    &\leq 2(d^n\cdot \|\bd{\mcal{T}}_t - \bd{\mcal{T}}^*\|_{\infty}^2 + |\epsilon_t|^2) \|\mcal{P}_{\bb{T}_t}\bd{\mcal{X}}_t\|_F^2\\
    &\leq 2 (d^n\cdot(\|\bd{\mcal{T}}_t\|_{\infty} + \|\bd{\mcal{T}}^*\|_{\infty})^2 + |\epsilon_t|^2)\cdot \overline{\optr{dof}}\cdot \mu_0^2\\
    &\leq C_n (\nu_0^2 r_{\min} \lambda_{\max}^2 + |\epsilon_t|^2) \cdot\overline{\optr{dof}}\cdot \mu_0^2.
\end{aligned}
\end{equation}
Meanwhile
\begin{equation}\label{equ2: E P_T_t G_t}
\begin{aligned}
\bb{E}_{t}\|\mcal{P}_{\bb{T}_t}\bd{\mcal{G}}_t\|_F^2 &= \bb{E}_{t}\la \bd{\mcal{X}}_t, \bd{\mcal{T}}_t - \bd{\mcal{T}}^*\ra^2 \|\mcal{P}_{\bb{T}_t}\bd{\mcal{X}}_t\|_F^2 + \sigma^2 \bb{E}_t\|\mcal{P}_{\bb{T}_t}\bd{\mcal{X}}_t\|_F^2\\
&\leq \frac{1}{d^{2n}} \sum_{\omega} d^{2n} \la\mcal{E}_{\omega}, \bd{\mcal{T}}_t - \bd{\mcal{T}}^* \ra^2 d^{2n} \|\mcal{P}_{\bb{T}_t}\bd{\mcal{E}}_{\omega}\|_F^2 + \overline{\optr{dof}}\cdot \sigma^2\\
&\leq d^{2n} \max_{\omega} \|\mcal{P}_{\bb{T}_t}\bd{\mcal{E}}_{\omega}\|_F^2\cdot\|\bd{\mcal{T}}_t - \bd{\mcal{T}}^*\|_F^2 + \overline{\optr{dof}}\cdot \sigma^2\\
&\leq \mu_0^2 \cdot \overline{\optr{dof}} \cdot\|\bd{\mcal{T}}_t - \bd{\mcal{T}}^*\|_F^2 + \overline{\optr{dof}}\cdot \sigma^2.
\end{aligned}
\end{equation}

\textit{Bounding $ \|\bd{\mcal{T}}_{t+1} - \bd{\mcal{T}}^*\|_F^2.$} We first estimate $\|\bd{\mcal{T}}_t^{+} - \bd{\mcal{T}}^*\|_F$. With the event $\bd{\mcal{E}}_t$ and the initialization condition, we have
$$
\|\bd{\mcal{T}}_t - \bd{\mcal{T}}^*\|_F^2 \leq 2\cdot\|\bd{\mcal{T}}_0 - \bd{\mcal{T}}^*\|_F^2  + 10 \eta \overline{\optr{dof}} \sigma^2\lesssim n^{-2}\eta^2 \lambda_{\min}^2
$$
as long as $(\lambda_{\min}/\sigma)^2 \gtrsim n^2 \eta^{-1} \overline{\optr{dof}}$. And we have 
$$
\begin{aligned}
\|\mcal{P}_{\bb{T}_t}\bd{\mcal{G}}_t\|_F^2&\leq 2(d^{2n}\cdot \|\bd{\mcal{T}}_t - \bd{\mcal{T}}^*\|_{\infty}^2 +|\epsilon_t|^2)\cdot \overline{\optr{dof}} \cdot \mu_0^2\\
    &\leq 2(n^2\mu_0^2 r_{\max}^2 d^2 \|\bd{\mcal{T}}_t - \bd{\mcal{T}}^*\|_F^2 + |\epsilon_t|^2)\cdot \overline{\optr{dof}} \cdot \mu_0^2\\
    &\lesssim (\eta^2 \lambda_{\min}^2 \cdot\mu_0^2 r_{\max}^2 d^2 + \sigma^2 \log d)\overline{\optr{dof}} \cdot \mu_0^2\\
    &\lesssim n^{-2}\lambda_{\min}^2 
\end{aligned}
$$
as long as $(\lambda_{\min}/\sigma)^2\gtrsim n^2\overline{\optr{dof}} \mu_0^2 \log d$ and $\eta \cdot n \mu_0^2 r_{\max}d(\overline{\optr{dof}})^{\frac{1}{2}} \lesssim 1$, and the second equality follows by \eqref{equ2: l_infty norm T_t-T^*} and \eqref{equ2: P_T_t G_t}. Thus we have $\|\bd{\mcal{T}}_t^{+} - \bd{\mcal{T}}^*\|_F\leq \|\bd{\mcal{T}}_t - \bd{\mcal{T}}^*\|_F +\eta \|\mcal{P}_{\bb{T}_t}\bd{\mcal{G}}_t\|_F\leq \frac{1}{1200n}\eta \lambda_{\min}$. Also we have $\|\bd{\mcal{T}}^*\|_{\infty}\leq \frac{\nu}{d^n}\|\bd{\mcal{T}}^*\|_F\leq \frac{\nu}{d^n}(\|\bd{\mcal{T}}^{+}_t\|_F + \|\bd{\mcal{T}}^{+}_t - \bd{\mcal{T}}^*\|_F)\leq \frac{\nu}{d^n}(\|\bd{\mcal{T}}^{+}_t\|_F + \frac{1}{10}\|\bd{\mcal{T}}^*\|_F)$. So we have $\|\bd{\mcal{T}}^*\|_{\infty}\leq \frac{10 \|\bd{\mcal{T}}^{+}_t\|_F}{9d^n}\nu = \xi_t$ thus $\|\mcal{W}_t - \bd{\mcal{T}}^*\|_F\leq \|\bd{\mcal{T}}_t^{+} - \bd{\mcal{T}}^*\|_F$. By \cref{lemma: perturbation bound of TTSVD}, we have 
$$
\begin{aligned}
    \|\bd{\mcal{T}}_{t+1} - \bd{\mcal{T}}^*\|_F^2&\leq \|\bd{\mcal{W}}_t - \bd{\mcal{T}}^*\|_F^2 + \frac{600n}{\lambda_{\min}} \|\bd{\mcal{W}}_t - \bd{\mcal{T}}^*\|_F^3\\
    &\leq\|\bd{\mcal{T}}^{+}_t - \bd{\mcal{T}}^*\|_F^2\left(1 + \frac{600n}{\lambda_{\min}}\|\bd{\mcal{T}}^{+}_t - \bd{\mcal{T}}^*\|_F\right)\\
    &\leq (1+\frac{\eta}{2})\|\bd{\mcal{T}}^{+}_t - \bd{\mcal{T}}^*\|_F^2.
\end{aligned}
$$
Now we consider $\bb{E}_t \|\bd{\mcal{T}}^{+}_t - \bd{\mcal{T}}^*\|_F^2$:
$$
\bb{E}_{t}\|\bd{\mcal{T}}_t^{+} - \bd{\mcal{T}}^*\|_F^2 = \bb{E}_t \left(\|\bd{\mcal{T}}_t - \bd{\mcal{T}}^*\|_F^2 - 2\eta \la\bd{\mcal{T}}_t - \bd{\mcal{T}}^*, \mcal{P}_{\bb{T}_t}\bd{\mcal{G}}_t \ra + \eta^2\|\mcal{P}_{\bb{T}_t}\bd{\mcal{G}}_t\|_F^2 \right).
$$
Since
$$
\begin{aligned}
    \bb{E}_t [\la\bd{\mcal{T}}_t - \bd{\mcal{T}}^*, \mcal{P}_{\bb{T}_t}\bd{\mcal{G}}_t \ra] &= \bb{E}_t [\la\mcal{P}_{\bb{T}_t}(\bd{\mcal{T}}_t - \bd{\mcal{T}}^*), \bd{\mcal{G}}_t \ra]\\
    &= \bb{E}_t [(\la\bd{\mcal{T}}_t - \bd{\mcal{T}}^*, \bd{\mcal{X}}_t\ra - \epsilon_t) \la\mcal{P}_{\bb{T}_t}(\bd{\mcal{T}}_t - \bd{\mcal{T}}^*), \bd{\mcal{X}}_t \ra]\\
    &=\frac{1}{d^{2n}} \sum_{\omega} d^n[\bd{\mcal{T}}_t - \bd{\mcal{T}}^*]_{\omega} d^n [\mcal{P}_{\bb{T}_t}(\bd{\mcal{T}}_t - \bd{\mcal{T}}^*)]_{\omega}\\
    &=\la \mcal{P}_{\bb{T}_t}(\bd{\mcal{T}}_t - \bd{\mcal{T}}^*),\bd{\mcal{T}}_t - \bd{\mcal{T}}^*\ra\\
    &=\|\mcal{P}_{\bb{T}_t}(\bd{\mcal{T}}_t - \bd{\mcal{T}}^*)\|_F^2.
\end{aligned}
$$
where in the second equality we use $\bb{E}_t(\epsilon_t \bd{\mcal{X}}_t) = \frac{1}{d^{2n}}\sum_{\omega}\bb{E}_{t}[\epsilon_{\omega}|\bd{\mcal{E}}_{\omega}]\bd{\mcal{E}}_{\omega} = 0$. Then we have
\begin{equation}\label{equ2: E T+-T*}
\begin{aligned}
    \bb{E}_{t}\|\bd{\mcal{T}}_t^{+} - \bd{\mcal{T}}^*\|_F^2 &= \|\bd{\mcal{T}}_t - \bd{\mcal{T}}^*\|_F^2 - 2\eta\|\mcal{P}_{\bb{T}_t}(\bd{\mcal{T}}_t - \bd{\mcal{T}}^*)\|_F^2 + \eta^2 \bb{E}_t\|\mcal{P}_{\bb{T}_t}\bd{\mcal{G}}_t\|_F^2\\
    &= (1-2\eta)\|\bd{\mcal{T}}_t - \bd{\mcal{T}}^*\|_F^2 + 2\eta\|\mcal{P}_{\bb{T}_t}^\perp(\bd{\mcal{T}}_t - \bd{\mcal{T}}^*)\|_F^2 + \eta^2 \bb{E}_t\|\mcal{P}_{\bb{T}_t}\bd{\mcal{G}}_t\|_F^2\\
    &\leq (1-\frac{3}{2}\eta) \|\bd{\mcal{T}}_t - \bd{\mcal{T}}^*\|_F^2 + \eta^2\bb{E}_t\|\mcal{P}_{\bb{T}_t}\bd{\mcal{G}}_t\|_F^2.
\end{aligned}
\end{equation}
where in the last inequality we use \cref{lemm: orth_comp tsp}. From \eqref{equ2: E P_T_t G_t} and \eqref{equ2: E T+-T*}, we have
$$
\begin{aligned}
    \|\bd{\mcal{T}}_{t+1} - \bd{\mcal{T}}^*\|_F^2 &\leq (1+\frac{\eta}{2})\bb{E}_t[\|\bd{\mcal{T}}_t^{+} - \bd{\mcal{T}}^*\|_F^2] + (1+\frac{\eta}{2})\left(\|\bd{\mcal{T}}_t^{+} - \bd{\mcal{T}}^*\|_F^2 - \bb{E}_t \|\bd{\mcal{T}}_t^{+} - \bd{\mcal{T}}^*\|_F^2\right)\\
    &\leq (1+\frac{\eta}{2})(1-\frac{3\eta}{2})\|\bd{\mcal{T}}_t - \bd{\mcal{T}}^*\|_F^2 + (1+\frac{\eta}{2})\eta^2 \bb{E}_t\|\mcal{P}_{\bb{T}_t}\bd{\mcal{G}}_t\|_F^2 \\
    &\qquad +(1+\frac{\eta}{2})\left(\|\bd{\mcal{T}}_t^{+} - \bd{\mcal{T}}^*\|_F^2 - \bb{E}_t\|\bd{\mcal{T}}_t^{+} - \bd{\mcal{T}}^*\|_F^2\right)\\
    &\leq (1-\eta) \|\bd{\mcal{T}}_t - \bd{\mcal{T}}^*\|_F^2 + 2\eta^2 \overline{\optr{dof}} \cdot \sigma^2 + \frac{\eta}{2} \|\bd{\mcal{T}}_t - \bd{\mcal{T}}^*\|_F^2\\
    &\qquad + (1+\frac{\eta}{2})\left(\|\bd{\mcal{T}}_t^{+} - \bd{\mcal{T}}^*\|_F^2 - \bb{E}_t\|\bd{\mcal{T}}_t^{+} - \bd{\mcal{T}}^*\|_F^2\right)\\
    &\leq (1-\frac{\eta}{2}) \|\bd{\mcal{T}}_t - \bd{\mcal{T}}^*\|_F^2 + 2\eta^2 \overline{\optr{dof}} \cdot \sigma^2 +(1+\frac{\eta}{2})\left(\|\bd{\mcal{T}}_t^{+} - \bd{\mcal{T}}^*\|_F^2 - \bb{E}_t\|\bd{\mcal{T}}_t^{+} - \bd{\mcal{T}}^*\|_F^2\right)
\end{aligned}
$$
where the second inequality follows from \eqref{equ2: E T+-T*} and the third inequality follows from \eqref{equ2: E P_T_t G_t} holds as long as $\eta \cdot \mu_0^2 \overline{\optr{dof}}\lesssim 1$. Telescoping this inequality and we get
\begin{equation}\label{equ2: tele sum}
\begin{aligned}    
    \|\bd{\mcal{T}}_{t+1} - \bd{\mcal{T}}^*\|_F^2&\leq (1-\frac{\eta}{2})^{t+1}\|\bd{\mcal{T}}_0 - \bd{\mcal{T}}^*\|_F^2 + 4\eta \overline{\optr{dof}} \sigma^2
    \\&+\sum_{l=0}^t \underbrace{(1-\frac{\eta}{2})^{t-l}(1+\frac{\eta}{2})[\|\bd{\mcal{T}}_l^{+} - \bd{\mcal{T}}^*\|_F^2 - \bb{E}_l \|\bd{\mcal{T}}_l^{+} - \bd{\mcal{T}}^*\|_F^2]}_{=:D_l}
\end{aligned}
\end{equation}

\textit{Estimation of $D_l$.} We first get the uniform bound of $\|\mcal{P}_{\bb{T}_l}\bd{\mcal{G}}_l\|_F^2$
\begin{equation}\label{equ2: uniform P T_l G_l}
\begin{aligned}
    \|\mcal{P}_{\bb{T}_l} \bd{\mcal{G}}_l\|_F^2 &\leq 2(\la\bd{\mcal{X}}_l, \bd{\mcal{T}}_l - \bd{\mcal{T}}^*\ra^2 + |\epsilon_l|^2)\|\mcal{P}_{\bb{T}_l}\bd{\mcal{X}}_l\|_F^2\\
    &\leq (2 d^{2n} \|\bd{\mcal{T}}_l - \bd{\mcal{T}}^*\|_{\infty}^2 + 2|\epsilon_l|^2) d^{2n} \max_{\omega}\|\mcal{P}_{\bb{T}_l}\bd{\mcal{E}}_{\omega}\|_F^2\\
    &\leq 2(n^2\mu_0^2 r_{\max}^2 d^2 \|\bd{\mcal{T}}_t - \bd{\mcal{T}}^*\|_F^2 + |\epsilon_t|^2)\cdot \overline{\optr{dof}} \cdot \mu_0^2.
\end{aligned}
\end{equation}
Meanwhile,
$$
\begin{aligned}
\left|\|\bd{\mcal{T}}^{+}_l - \bd{\mcal{T}}^*\|_F^2 - \bb{E}_{l}\|\bd{\mcal{T}}^{+}_l - \bd{\mcal{T}}^*\|_F^2\right|&\leq 2\eta \left|\la\bd{\mcal{T}}_l - \bd{\mcal{T}}^*, \mcal{P}_{\bb{T}_l}\bd{\mcal{G}}_l \ra - \bb{E}_l \la\bd{\mcal{T}}_l - \bd{\mcal{T}}^*, \mcal{P}_{\bb{T}_l}\bd{\mcal{G}}_l \ra \right|\\
&+\eta^2 \left|\|\mcal{P}_{\bb{T}_l} \bd{\mcal{G}}_l\|_F^2 - \bb{E}_l \|\mcal{P}_{\bb{T}_l} \bd{\mcal{G}}_l\|_F^2 \right|
\end{aligned}
$$
We now consider the uniform bound for $\left|\la\bd{\mcal{T}}_l - \bd{\mcal{T}}^*, \mcal{P}_{\bb{T}_l}\bd{\mcal{G}}_l \ra\right|$. Using the Cauchy-Schwarz inequality and \eqref{equ2: uniform P T_l G_l},
$$
\begin{aligned}
    \left|\la\bd{\mcal{T}}_l - \bd{\mcal{T}}^*, \mcal{P}_{\bb{T}_l}\bd{\mcal{G}}_l \ra\right|&\leq \|\bd{\mcal{T}}_l - \bd{\mcal{T}}^*\|_F\|\mcal{P}_{\bb{T}_l} \bd{\mcal{G}}_l\|_F\\
    &\lesssim n \mu_0^2 r_{\max}d(\overline{\optr{dof}})^{1/2} \|\bd{\mcal{T}}_l - \bd{\mcal{T}}^*\|_F^2 + \mu_0 (\overline{\optr{dof}})^{1/2}|\epsilon_l|\cdot \|\bd{\mcal{T}}_l - \bd{\mcal{T}}^*\|_F
\end{aligned}
$$
Thus, we have the uniform bound for $D_l$
\begin{equation}
\begin{aligned}
    |D_l|&\lesssim (1-\frac{\eta}{2})^{t-l} \cdot \eta (|\la\bd{\mcal{T}}_l -\bd{\mcal{T}}^*, \mcal{P}_{\bb{T}_l} \bd{\mcal{G}}_l \ra - \bb{E}_l \la\bd{\mcal{T}}_l-\bd{\mcal{T}}^*, \mcal{P}_{\bb{T}_l}\bd{\mcal{G}}_l\ra|)\\
    &\quad +(1-\frac{\eta}{2})^{t-l}\eta^2 (|\|\mcal{P}_{\bb{T}_l}\bd{\mcal{G}}_l\|_F^2 - \bb{E}_l\|\mcal{P}_{\bb{T}_l}\bd{\mcal{G}}_l\|_F^2|)\\
    &\lesssim (1-\frac{\eta}{2})^{t-l}\cdot \eta \left(n \mu_0^2 r_{\max}d(\overline{\optr{dof}})^{1/2} \|\bd{\mcal{T}}_l - \bd{\mcal{T}}^*\|_F^2 + \mu_0 (\overline{\optr{dof}})^{1/2}|\epsilon_l|\cdot \|\bd{\mcal{T}}_l - \bd{\mcal{T}}^*\|_F \right)\\
    &\quad+ (1-\frac{\eta}{2})^{t-l} \eta^2 \left(n^2\mu_0^2 r_{\max}^2 d^2 \|\bd{\mcal{T}}_t - \bd{\mcal{T}}^*\|_F^2 + |\epsilon_t|^2\right)\cdot \overline{\optr{dof}} \cdot \mu_0^2\\
    &\lesssim (1-\frac{\eta}{2})^{t-l}\eta\cdot n \mu_0^2 d^2 r_{\max}^2 \|\bd{\mcal{T}}_l - \bd{\mcal{T}}^*\|_F^2 + \eta\sigma^2\log d\\
    &\leq \frac{1}{2}(1-\frac{\eta}{4})^{t+1} \|\bd{\mcal{T}}_0 - \bd{\mcal{T}}^*\|_F^2(\log d)^{-1} + \eta \overline{\optr{dof}} \cdot \sigma^2(\log d)^{-1}
\end{aligned}
\end{equation}
where the third inequality uses $2ab\leq a^2 + b^2$ and $\eta\cdot n \mu_0^2 \overline{\optr{dof}}\lesssim 1$, and the last inequality is from the event $\bd{\mcal{E}}_t$ that $\|\bd{\mcal{T}}_l - 
\bd{\mcal{T}}^*\|_F^2\leq 2(1-\frac{\eta}{4})^l \|\bd{\mcal{T}}_0 - \mcal{T^*}\|_F^2 + 10 \eta \overline{\optr{dof}} \sigma^2$ and $\eta\cdot n \mu_0^2 r_{\max}^2 d^2 \log d\lesssim 1$.

We now consider the variance bound for $\mcal{D}_l$. Firstly, we consider $\bb{E}_{l}\|\mcal{P}_{\bb{T}_l}\bd{\mcal{G}}_l\|_F^4$. By \eqref{equ2: l_infty norm T_t-T^*} and \eqref{equ2: E_t P_{T_t} X_t}, we have
$$
\begin{aligned}
    \bb{E}_l \|\mcal{P}_{\bb{T}_l}\bd{\mcal{G}}_l\|_F^4&\leq 8\cdot\bb{E}_l \left[\la \bd{\mcal{X}}_l,\bd{\mcal{T}}_l - \bd{\mcal{T}}^*\ra^4 \|\mcal{P}_{\bb{T}_l}\bd{\mcal{X}}_l\|_F^4\right] + 8\cdot\bb{E}_l \left[\epsilon_{l}^4\|\mcal{P}_{\bb{T}_l}\bd{\mcal{X}}_l\|_F^4\right]\\
    &\leq \bb{E}_l [\la \bd{\mcal{T}}_l -\bd{\mcal{T}}^*, \bd{\mcal{X}}_l\ra^2]\|\bd{\mcal{T}}_l - \bd{\mcal{T}}^*\|_{\infty}^2 d^{2n} \|\mcal{P}_{\bb{T}_l}\bd{\mcal{X}}_l\|_F^4 + \sigma^4 \log^2 (d) \|\mcal{P}_{\bb{T}_l}\bd{\mcal{X}}_l\|_F^2 \bb{E}_l\|\mcal{P}_{\bb{T}_l}\bd{\mcal{X}}_l\|_F^2\\
    &\lesssim n^2 \cdot \mu_0^6 r_{\max}^2 d^2 (\overline{\optr{dof}})^2 \|\bd{\mcal{T}}_l - \bd{\mcal{T}}^*\|_F^4 +  \sigma^4 \log^2 (d) (\overline{\optr{dof}})^2 \mu_0^2 .
\end{aligned}
$$
On the other hand,
$$
\begin{aligned}
    \bb{E}_l\left|\la\bd{\mcal{T}}_l - \bd{\mcal{T}}^*, \mcal{P}_{\bb{T}_l}\bd{\mcal{G}}_l \ra\right|^2 &= \bb{E}_l\left|\la\bd{\mcal{T}}_l - \bd{\mcal{T}}^*, \mcal{P}_{\bb{T}_l}\bd{\mcal{X}}_l \ra\right|^2(\epsilon_l^2 + |\la\bd{\mcal{T}}_l - \bd{\mcal{T}}^*, \bd{\mcal{X}}_l \ra|^2) \\
    &\leq 2\sigma^2 \log (d) \cdot \bb{E}_{l}|\la\bd{\mcal{T}}_l - \bd{\mcal{T}}^*, \mcal{P}_{\bb{T}_l}\bd{\mcal{X}}_l \ra|^2 + d^{2n} \|\bd{\mcal{T}}_l - \bd{\mcal{T}}^*\|_{\infty}^2 \bb{E}_l |\la\bd{\mcal{T}}_l - \bd{\mcal{T}}^*, \mcal{P}_{\bb{T}_l}\bd{\mcal{X}}_l \ra|^2\\
    &\lesssim  \sigma^2 \log (d) \|\bd{\mcal{T}}_l - \bd{\mcal{T}}^*\|_F^2 + n^2 \mu_0^2 r_{\max}^2 d^2\|\bd{\mcal{T}}_l - \bd{\mcal{T}}^*\|_F^4.
\end{aligned}
$$
So as long as $\eta \cdot\mu_0 \overline{\optr{dof}}\lesssim 1$, we have
$$
\begin{aligned}
    \optr{Var}_l D_l&\lesssim  (1-\frac{\eta}{2})^{2t-2l} \eta^4 \bb{E}_l \|\mcal{P}_{\bb{T}_l}\bd{\mcal{G}}_l\|_F^4 + (1-\frac{\eta}{2})^{2t-2l} \eta^2 \bb{E}_l |\la \bd{\mcal{T}}_l - \bd{\mcal{T}}^*, \mcal{P}_{\bb{T}_l} \bd{\mcal{G}}_l\ra|^2 \\
    &\lesssim (1-\frac{\eta}{2})^{2t-2l} \eta^2 n^2 \mu_0^2 r_{\max}^2 d^2\|\bd{\mcal{T}}_l - \bd{\mcal{T}}^*\|_F^4 + (1-\frac{\eta}{2})^{2t-2l} \eta^2 \sigma^2 \log (d) \|\bd{\mcal{T}}_l - \bd{\mcal{T}}^*\|_F^2\\
    &\qquad + (1-\frac{\eta}{2})^{2t-2l}\eta^4 \sigma^4 \log^2 (d) \mu_0^2 (\overline{\optr{dof}})^2.
\end{aligned}
$$
Together with $2(1-\frac{1}{4} \eta)^l \|\bd{\mcal{T}}_0 - \mcal{T^*}\|_F^2 + 10 \eta \cdot\overline{\optr{dof}} \sigma^2$, we obtain the bound of summation
$$
\begin{aligned}
    \sum_{l=0}^t \optr{Var}_l D_l&\lesssim (1-\frac{\eta}{4})^{2t + 2} \eta n^2 \mu_0^2 r_{\max}^2 d^2 \|\bd{\mcal{T}}_0 - \bd{\mcal{T}}^*\|_F^2 + \eta^3 \mu_0^2 r_{\max}^2 d^2(\overline{\optr{dof}})^2\sigma^4 + \eta^3 \mu_0^2 (\overline{\optr{dof}})^2\sigma^4 \log^2 (d) \\
    &\qquad + (1-\frac{\eta}{4})^{t+1}\eta \sigma^2\log (d) \|\bd{\mcal{T}}_0 - \bd{\mcal{T}}^*\|_F^2 +\eta^2 \overline{\optr{dof}}\sigma^4 \\
    &\leq \frac{1}{4}(1-\frac{\eta}{4})^{2t+2}\|\bd{\mcal{T}}_0 - \bd{\mcal{T}}^*\|_F^4 \log^{-1 (d)}+\eta^2 (\overline{\optr{dof}})^2 \sigma^4 \log^{-1} (d),
\end{aligned}
$$
where the inequality holds as long as $\log d\lesssim \overline{\optr{dof}}$ and $\eta\cdot n^2\mu_0^2 r_{\max}^2 d^2 \log d \lesssim 1$. Using the variance bound and martingale inequality in \cref{thm: martingale concentration}, we see that with probability exceeding $1-d^{-20}$,
$$
\sum_{l=0}^t D_l\leq (1-\frac{\eta}{4})^{t+1} \|\bd{\mcal{T}}_0 - \bd{\mcal{T}}^*\|_F^2 + 2\eta \overline{\optr{dof}} \sigma^2.
$$
Plug this into \eqref{equ2: tele sum}, we have
\begin{equation}
    \begin{aligned}
        \|\bd{\mcal{T}}_{t+1} - \bd{\mcal{T}}^*\|_F^2\leq 2(1-\frac{\eta}{4})^{t+1}\|\bd{\mcal{T}}_0  - \bd{\mcal{T}}^*\|_F^2 + 6 \eta \overline{\optr{dof}}\sigma^2.
    \end{aligned}
\end{equation}
By \cref{lemma: trim incoh} and \cref{lemma: spiki imply incoh} we have the incoherence of $\bd{\mcal{T}}^{t+1}$ is bounded by $2\kappa_0^2 \nu = 2 \kappa_0 \sqrt{\mu}$, thus we conclude the proof.

\section*{Proof of \cref{prop: intialization estimation}.}\label{proof: lemma initialization estimation}

We define the following projection distance and the chordal distance between two orthogonal matrices $\bd{U}, \bd{V}\in \bb{R}^{p\times r}$ as
$$
d_p(\bd{U}, \bd{V}) := \|\bd{U} \bd{U}^\top-\bd{V}\bd{V}^\top \|_F, \quad d_c(\bd{U}, \bd{V}) := \min_{\bd{Q}\in \bb{Q}_r}\|\bd{U} \bd{Q} - \bd{V}\|_F
$$
we have the property that $\frac{1}{\sqrt{2}}d_c(\bd{U}, \bd{V})\leq d_p (\bd{U}, \bd{V})\leq d_c(\bd{U}, \bd{V})$.

We denote $\bd{\mcal{T}}^* = [\bd{Z}_1^*, \bd{Z}_2^*, \bd{Z}_3^*]$ be the rank-$(r_{m_1}, r_{m_1 + m_2})$ decomposition of $\bd{\mcal{T}}^*$. Also, we denote the following matrix $R_1 = \argmin_{\bd{R}\in \bb{O}_{r_{m_1}}}\|\wht{\bd{Z}}^{\leq i} - \bd{Z}^{*\leq i} R\|_F$ (similarly for $\bd{R}_2$). For ease of exposition, we first consider the result for the noiseless case.

\textbf{The noiseless case:}
Consider $\|\bd{N}_1  - \bd{N}_1^*\|$, by \cref{lemma: noiseless eig estimater}, 
\begin{equation}\label{equ: noiseless N_1}
\|\bd{N}_1 - \bd{N}_1^*\|\leq C \frac{d^{n} n \log (d)}{K_1} r_{\min} \lambda_{\max}^2 \nu^2\left[\left(1+ d^{-\frac{2n}{3}}\right)^{1/2} + \frac{d^{n}}{K_1} + \left(\frac{K_1}{d^{n} n \log (d)}\right)^{1/2}\right]\leq \mcal{H}_1\cdot \lambda_{\max}^2,
\end{equation}
as long as $K_1\geq C  n d^{n}\log (d) r_{\min}  \nu^2\cdot  \mcal{H}_1^{-1} + C n d^{n} \log (d) r_{\min}^2 \nu^4\cdot \mcal{H}_1^{-2}$. Additionally,
$$
d_p(\wtd{\bd{Z}}_1, \bd{Z}_1^*)\leq \frac{2 \sqrt{r_{m_1}}}{\lambda_{\min}^2}\|\bd{N}_1 - \bd{N}_1^*\|\leq C \sqrt{r_{m_1}}\kappa_0^2 \mcal{H}_1.
$$
By \cref{lemma: Incoherence of matrix}, we have $\optr{Incoh}(\wht{\bd{Z}}_1)\leq \sqrt{3\mu}$ and $d_c(\wht{\bd{Z}}_1, \bd{Z}_1^*)\leq C \sqrt{r_{m_1}}\kappa_0^2 \mcal{H}_1.$ 
% and $d_c(\wht{Z}_1, Z_1^*)\leq C \sqrt{r_{m_1}}\kappa_0^2 \mcal{H}_1$ since $d_p(\widetilde{Z}_1, Z_1^*)\leq 2 \sqrt{r_{m_1}} \lambda_{\min}^{-2}\cdot \|N_1 - N_1^*\|$.
Consider $\wtd{\bd{Z}}_2$, by \cref{lemma: noisy truncated eig estimator} without noise we estimate that 
\begin{equation}\label{equ: noiseless N_2}
\begin{aligned}
    \|(\wht{\bd{Z}}_1\otimes \bd{I})^\top (\bd{N}_2 - \bd{N}_2^*)(\wht{\bd{Z}}_1\otimes \bd{I})\|&\leq C n^2 \log^2(d) \frac{\nu^2 r_{\min} \lambda_{\max}^2}{K_2}\left(\mu r_{m_1} d^{\frac{n}{3}} + \frac{\mu r_{m_1} d^{\frac{2n}{3}}}{K_2} +\frac{(\mu r_{m_1} K_2)^{\frac{1}{2}}}{n^{3/2}\log^{3/2}(d)} \right)\\
    &\leq \mcal{H}_2\cdot \lambda_{\max}^2,
\end{aligned}
\end{equation}
as long as $K_2 \geq C n^2 \nu^2 r_{\min} \mu r_{\max} d^{\frac{n}{3}} \log^2 d \cdot \mcal{H}_2^{-1} + C^2 n \log (d) \nu^4 r_{\min}^2 \mu r_{\max}\cdot \mcal{H}_2^{-2}$. On the other hand, we see that
$$
d_{p}(L(\wtd{\bd{Z}}_2), (\bd{R}_1\otimes \bd{I})^\top L(\bd{Z}_1^*))\leq 2\sqrt{r_{m_1}}\lambda_{\min}^{-2} \|(\wht{\bd{Z}}_1\otimes \bd{I})^\top (\bd{N}_2 - \bd{N}_2^*)(\wht{\bd{Z}}_1\otimes \bd{I})\| + 4 \sqrt{r_{m_1}}\kappa_0^2 d_c(\wht{\bd{Z}}_1, \bd{Z}_1^*).
$$
Thus we have
$$
\begin{aligned}
    d_c(\wht{\bd{Z}}^{\leq 2}, \bd{Z}^{*\leq 2})&\leq \sqrt{r_{m_1}}d_c(\wht{\bd{Z}}_1, \bd{Z}_1^*) + \sqrt{2}d_p(L(\wht{\bd{Z}}_2), (\bd{R}_1\otimes \bd{I})^\top L(\bd{Z}^*_2))\\
    &\leq \sqrt{r_{m_1}} (1+ 16 \pi \kappa_0^2)d_c(\wht{\bd{Z}}_1, \bd{Z}_1^*) + 8 \sqrt{2}\pi \sqrt{r_{\min}} \lambda_{\min}^{-2}\|(\wht{\bd{Z}}_1\otimes \bd{I})^\top (\bd{N}_2 - \bd{N}_2^*)(\wht{\bd{Z}}_1\otimes \bd{I})\|\\
    &\leq C r_{m_1} (1+ 16 \pi \kappa_0^2 )\kappa_0^{2} \cdot \mcal{H}_1 + 8 \sqrt{2}\pi \sqrt{r_{\min}} \kappa_0^{2}\cdot  \mcal{H}_2.
\end{aligned}
$$
Consider $\|\bd{\mcal{T}}^* - \wht{\bd{\mcal{Z}}}\|_F$. We first notice that $\optr{Incoh}(\wht{\bd{Z}}^{\leq 2})\leq 3\mu r_{\max}^{\frac{3}{2}}$ and by \cref{lemma: last truncation}
\begin{equation}\label{equ: noiseless N_3}
\begin{aligned}
\left\|\wht{\bd{Z}}^{\leq 2\top}\bd{\mcal{T}}^{*\la m_1 + m_2\ra}  - \wht{\bd{Z}}_3\right\|&\leq C \left(\frac{d^{\frac{n}{3}}n\mu r_{\max}^2 \sqrt{r_{\min}} \lambda_{\max} \nu\log (d)}{K_3} +\left(\frac{r_{\min} \lambda_{\max}^2 \nu^2 (\mu^2 r_{\max}^4 \vee d^{\frac{2n}{3}}) n \log(d)}{K_3}\right)^{\frac{1}{2}}\right)\\
&\leq \mcal{H}_3\cdot \lambda_{\max}
\end{aligned}
\end{equation}
as long as $K_3 \geq d^{\frac{n}{3}} n \nu \mu r_{\max}^2  \sqrt{r_{\min}} \log (d)\cdot \mcal{H}_3^{-1} + r_{\min} \nu^2 (\mu^2 r_{\max}^4 \vee d^{\frac{2n}{3}})n \log (d) \cdot \mcal{H}_3^{-2}$. On the other hand 
$$
\begin{aligned}
\|\bd{R}_2^\top \bd{Z}_3^* - \wht{\bd{Z}}_3\|_F &= \|(\bd{Z}^{*\leq 2} \bd{R}_{2})^\top \bd{\mcal{T}}^{*\la m_1 +m_2\ra} - \wht{\bd{Z}}_3\|_F\\
&\leq \|\wht{\bd{Z}}^{\leq 2\top}\bd{\mcal{T}}^{*\la m_1 + m_2\ra}  - \wht{\bd{Z}}_3\|_F + \|(\bd{Z}^{*\leq 2} \bd{R}_{2} - \wht{\bd{Z}}^{\leq 2})^\top \mcal{\bd{T}}^{*\la m_1 +m_2\ra} \|_F\\
&\leq \|\wht{\bd{Z}}^{\leq 2\top}\bd{\mcal{T}}^{*\la m_1 + m_2\ra}  - \wht{\bd{Z}}_3\|_F + d_c(\wht{\bd{Z}}^{\leq 2}, \bd{Z}^{*\leq 2})\cdot \lambda_{\max},
\end{aligned}
$$
thus we have the estimation of $\|\bd{\mcal{T}}^* - \wht{\mcal{\bd{Z}}}\|_F$
$$
\begin{aligned}
    \|\bd{\mcal{T}}^* - \wht{\bd{\mcal{Z}}}\|_F&\leq \|\bd{Z}^{*\leq 2} \bd{R}_2 \bd{R}_2^\top \bd{Z}_3^* - \wht{\bd{Z}}^{\leq 2}\wht{\bd{Z}}_3\|_F\\
    &\leq \|(\bd{Z}^{*\leq 2} \bd{R}_{2} - \wht{\bd{Z}}^{\leq 2})  \bd{R}_2^\top \bd{Z}_3^* \|_F + \|\bd{R}_2^\top \bd{Z}_3^* - \wht{\bd{Z}}_3\|_F\\
    &\leq 2 \lambda_{\max} d_c(\wht{\bd{Z}}^{\leq 2}, \bd{Z}^{*\leq 2}) + \mcal{H}_3\cdot \lambda_{\max}\\
    &\leq C (r_{\max}\kappa_0^4 \cdot \mcal{H}_1 +  \sqrt{r_{\min}} \kappa_0^2\cdot \mcal{H}_2 + \mcal{H}_3) \cdot \lambda_{\max}.
\end{aligned}
$$
Thus, if the requirements on the sample number $K_1, K_2, K_3$ satisfies that
$$
\begin{aligned}
    K_1&\geq C n d^{n}\log (d) r_{\min} r_{\max} \nu^2\kappa_0^4 \lambda_{\max}  \scr{K}^{-1} + C n d^{n}\log (d) r_{\min}^2 r_{\max}^2 \nu^4 \kappa_0^8 \lambda_{\max}^2 \scr{K}^{-2}\\
    K_2&\geq C n^2 d^{\frac{n}{3}}\log^2(d)r_{\min}^{\frac{3}{2}}r_{\max} \mu \nu^2 \kappa_0^2 \lambda_{\max} \scr{K}^{-1} + C n \log(d) r_{\min}^3 r_{\max} \mu \nu^4 \kappa_0^4 \lambda_{\max}^2 \scr{K}^{-2}\\
    K_3&\geq C n d^{\frac{n}{3}} \log(d) r_{\min}^{\frac{1}{2}} r_{\max}^2 \mu \nu \lambda_{\max} \scr{K}^{-1} + Cn\log(d) (d^{\frac{2n}{3}}\vee \mu^2 r_{\max}^4) r_{\min}  \nu^2 \lambda_{\max}^2\scr{K}^{-2}.
\end{aligned}
$$
Then with \cref{lemma: trim incoh} we can conclude that 
$$
\|\bd{\mcal{T}}_0-\bd{\mcal{T}}^*\|_F\leq \scr{K}_{n,d, r, \nu, \mu , \kappa_0}\quad \text{ and }\quad \optr{Incoh}(\bd{\mcal{T}}_0)\leq 2\kappa_0^2 \nu.
$$

\textbf{The noisy case:}
The procedure is similar to the noiseless case, and we need to obtain the results in \eqref{equ: noiseless N_1}, \eqref{equ: noiseless N_2} and \eqref{equ: noiseless N_3} considering the noise. For the $\|\bd{N}_1 - \bd{N}_1^*\|$ we use the following lemma:
\begin{lemma}[Theorem 2 \cite{xia2021statistically}] \label{lemma: noisy eig estimater}
There exists absolute constant $C_1, C_2$ such that for any $\alpha\geq 1$, if 
$$
K_1\geq C_1 \alpha d^n n^2 \log(d)
$$
then with probability exceeding $1-d^{-2\alpha}$
$$
\begin{aligned}
\|\bd{N}_1 - \bd{N}_1^*\|\leq &C_2 \left( (\frac{\sigma}{d^{n}} + \|\bd{\mcal{T}}^*\|_{\infty})\|\bd{\mcal{T}}^*\|_F\sqrt{\frac{\alpha n d^{\frac{8n}{3}} \log(d)}{K_1}}\right. \\
&\left. + \alpha^3 \left(\frac{\sigma^2}{d^{2n}} + \|\bd{\mcal{T}}^*\|_{\infty}^2 n^2\log^2(d)\right)\frac{n^{3} d^{3n}\log^3 (d)}{K_1} \left(1+ d^{-\frac{n}{3}}\right)\right)
\end{aligned}
$$
\end{lemma}
The noisy term in \cref{lemma: noisy eig estimater} differs slightly from that in Theorem 2 of \cite{xia2021statistically}. In \cite{xia2021statistically}, each noisy observation is modeled as $\mcal{T}^*(\omega_t) + \xi_t$, so that the standard deviation satisfies $\sigma(\xi_t) = \frac{1}{d^n} \sigma(\epsilon_t) = \frac{1}{d^n} \sigma$. Thus the \eqref{equ: noiseless N_1} still holds with probability exceeding $1- d^{-200}$ as long as the sample size $K_1$ satisfies
$$
K_1\geq C n^{5}d^n \log^5 (d) r_{\min}\nu^2 \cdot \mcal{H}_1^{-1} + C n d^{\frac{2n}{3}} \log (d) r_{\min}^2 \nu^2 \cdot \mcal{H}_1^{-2},
$$
and the signal-to-noise-ratio (SNR) condition satisfies
$$
\frac{\lambda_{\min}}{\sigma}\geq \wht{C} \left(\frac{r_{\min} n d^{\frac{2n}{3}} \log (d)}{K_1} \right)^{\frac{1}{2}} \cdot {\mcal{H}_1^{-1}} + \wht{C}\left(\frac{n^{3}d^n \log^3 (d)}{K_1}\right)^{\frac{1}{2}}\cdot {\mcal{H}_1^{-\frac{1}{2}}}.
$$
for some constant $C, \wht{C}>0$. While for \eqref{equ: noiseless N_2}, we use the \cref{lemma: noisy truncated eig estimator}. The \eqref{equ: noiseless N_2} holds under the sample size condition and the SNR condition
$$
\begin{aligned}
\frac{\lambda_{\min}}{\sigma}&\geq \wht{C} \sqrt{r_{\min}}\nu \left(\frac{\mu r_{\max} n \log (d)}{K_2} + \frac{\mu r_{\max} d^{\frac{2n}{3}} n^2 \log^2 (d)}{K_2^2} + \sqrt{\frac{\mu r_{\max} n \log (d)}{K_2}} \right)\cdot {\mcal{H}_2^{-1}}\\
&\quad + \wht{C} \left(\frac{\mu r_{\max} n^\frac{3}{2}\log^{\frac{3}{2}}(d)}{K_2} + \frac{\mu r_{\max} n^{\frac{5}{2}}\log^{\frac{5}{2}}(d)d^{\frac{2n}{3}}}{K_2^2}\right)^{\frac{1}{2}}\cdot {\mcal{H}_2^{-\frac{1}{2}}}.
\end{aligned}
$$
Then we need to bound the last term $\left\|\wht{\bd{Z}}^{\leq 2\top}\bd{\mcal{T}}^{*\la m_1 + m_2\ra}  - \wht{\bd{Z}}_3\right\|$, by \cref{lemma: noisy truncated eig estimator}, the \eqref{equ: noiseless N_3} still holds under the SNR condition,
$$
\frac{\lambda_{\min}}{\sigma}\geq \wht{C}\left(\sqrt{\frac{d^{\frac{2n}{3}} n \log (d)}{K_3}} + \frac{d^n \mu r_{\max} n \log(d)}{K_3} \right) \cdot {\mcal{H}_3^{-1}}
$$

Thus, as long as the sample number $K_1, K_2$ and $K_3$ have a lower bound $K$ satisfies
$$
K\geq C n^5 d^n \log^5(d) r_{\min}^{\frac{1}{2}}r_{\max}^2 \mu \nu^2 \kappa_0^4 \lambda_{\max} \cdot \scr{H}^{-1} + C nd^n \log (d) r_{\min} r_{\max}^4 \nu^4 \mu \kappa_0^8 \lambda_{\max}^2 \cdot \scr{H}^{-2}
$$
and the SNR condition satisfies that
$$
\frac{\lambda_{\min}}{\sigma}\geq \wht{C}\frac{n d^n  \log(d) \mu r_{\max}^2 \kappa_0^4 \lambda_{\max}}{K} \cdot\scr{H}^{-1} + \wht{C}\left(\frac{n^3 d^n \log^3(d) r_{\max} \kappa_0^4 \lambda_{\max}}{K} \right)^{\frac{1}{2}}\cdot \scr{H}^{-\frac{1}{2}}
$$
\eqref{equ: initial estimate} holds and we conclude the proof of \cref{prop: intialization estimation}.

\section{Conclusion and Discussion}

In this paper, we developed a provably convergent online Riemannian gradient descent algorithm for quantum state tomography of matrix product density operators. By characterizing the Hermitian property within the MPO representation, we showed that QST under Pauli or generalized Gell-Mann measurements can be reformulated as a noisy low TT-rank tensor completion problem over a real tensor manifold. Based on this formulation, we established the local linear convergence of the proposed oRGD algorithm with a measurement setting complexity that scales quadratically with the system size, up to rank-dependent and logarithmic factors. To meet the initialization requirement, we also proposed a tailored spectral initialization algorithm and provided its theoretical guarantee. Numerical experiments on several representative quantum states further demonstrate the efficiency and scalability of the proposed method.

Several directions remain for future investigation. One important extension is to incorporate the positive semidefinite and trace-one constraints of density matrices more explicitly into the optimization procedure. Another promising direction is to study adaptive or active measurement strategies that select informative tensor-product observables during the online reconstruction process. It would also be interesting to extend the present framework to more general tensor-network representations beyond one-dimensional MPOs, thereby enabling provable and scalable tomography for a broader class of many-body quantum systems.

\section*{Acknowledgments}
We thank the SJTU supercomputer center for providing the computing services. This work was supported by the National Natural Science Foundation of China (Grant No. 125B2026). This work was supported by Hong Kong Research Grant Council (HKRGC) GRF grants 16307023, 16306124, and 16307325.

\bibliography{ref}
\bibliographystyle{IEEEtran}

\end{document}